\documentclass[article]{IEEEtran}
\onecolumn
\IEEEoverridecommandlockouts

\usepackage{macros}

\begin{document}

\title{Pufferfish Privacy: \\ An Information-Theoretic Study}

\author{Theshani~Nuradha,~\IEEEmembership{Student~Member,~IEEE,}
        and~Ziv~Goldfeld,~\IEEEmembership{Member,~IEEE}
        
\thanks{This paper was presented, in part, at the 2022 IEEE International Symposium on Information Theory \cite{nuradha2022PP}. The work of Z. Goldfeld is partially supported by NSF grants CCF-1947801,  CCF-2046018, and DMS-2210368, and the 2020 IBM Academic Award.}        
        }

\maketitle
\begin{abstract}
Pufferfish privacy (PP) is a generalization of differential privacy (DP), that offers flexibility in specifying sensitive information and integrates domain knowledge into the privacy definition. Inspired by the illuminating formulation of DP in terms of mutual information due to Cuff and Yu, this work explores PP through the lens of information theory. We provide an information-theoretic formulation of PP, termed mutual information PP (MI PP), in terms of the conditional mutual information between the mechanism and the secret, given the public information. We show that MI PP is implied by the regular PP and characterize conditions under which the reverse implication is also true, recovering the relationship between DP and its information-theoretic variant as a special case. We establish convexity, composability, and post-processing properties for MI PP mechanisms and derive noise levels for the Gaussian and Laplace mechanisms. The obtained mechanisms are applicable under relaxed assumptions and provide improved noise levels in some regimes. 
Lastly, applications to auditing privacy frameworks, statistical inference tasks, and algorithm stability are explored.

\end{abstract}

\section{Introduction}

With the exponential increase in personal data shared online and recent advancements in data mining techniques, privacy concerns have become more pressing than ever. Statistical privacy frameworks seek to address these threats in a principled manner subject to formal guarantees \cite{census_haney2017utility}. Differential privacy (DP)~\cite{DMNS06} is a popular framework, which preserves the privacy of individual records while enabling aggregate queries about a database. However, DP only deals with one type of private information (individual records modeled by rows of the database) and does not allow to incorporate domain knowledge into the framework. To address these limitations, a versatile generalization of DP called Pufferfish Privacy (PP) was proposed in~\cite{KM14}. PP enables customization of what constitutes private information and explicitly integrates distributional assumptions into its definition. Nevertheless, the flexibility of the PP definition comes at a cost as the general framework is hard to work with and derive mechanisms for. This work aims to circumvent this impasse by proposing a new structured PP framework along with a natural information-theoretic formulation thereof, which lends well for analysis and enables devising mechanisms and exploring various additional~applications.

\begin{wrapfigure}{R}{0.5\textwidth}
\vspace{-2mm}
{\centering
\includegraphics[width=0.48\textwidth]{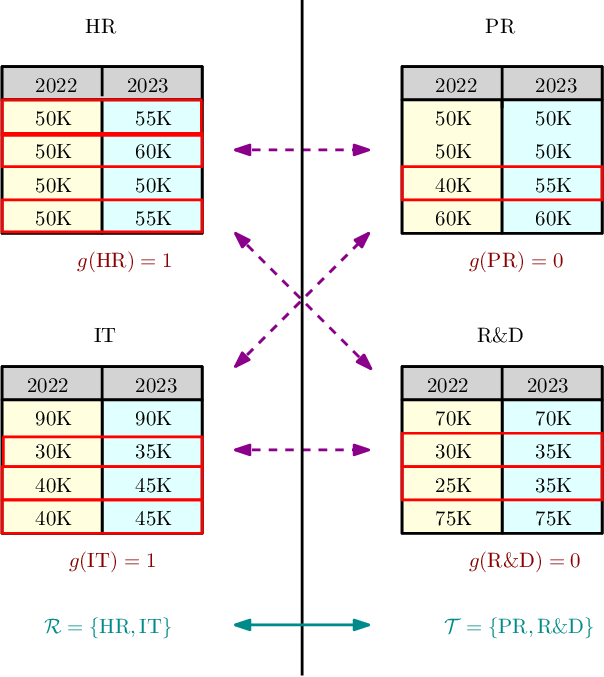} 
    \caption{2022-2033 salary data in four departments: HR, IT, PR, R\&D. The goal is to publish the average 2023 salary in each department (the average of the blue cells) while hiding whether the number raises (marked by red frames) is $\leq 2$ corresponding to $g(\cdot)=0$ or $>2$ corresponding to $g(\cdot)=1$. The average 2022 salaries (yellow cells) are public knowledge.}\label{fig:structured PP example}
 }
\end{wrapfigure}

 \subsection{Pufferfish Privacy}\label{subsec:PP_intro}
Consider salary data from 2022-2023 at a company with four departments: HR, IT, PR, and R\&D. The company wants to publish the average 2023 salary in each department while concealing whether more or less than $m$ employees got a raise.
The average salaries from 2022 are publicly available. More formally, the goal is to publish $f(x)=\frac{1}{n}\sum_{i=1}^nx(i,2023)$ while privatizing whether
$g(x)=\mathds{1}_{\cA_m}$, where $\cA_m=\big\{ |\{i:\,x(i,2022)<x(i,2023)\}| > m  \big\}$,
$x\in\cX\coloneqq \{\mathrm{HR},\mathrm{IT},\mathrm{PR},\mathrm{R\&D}\}$, and $x(i,j)$ is the salary of the $i$th employee during year $j$ in department $x$. The average salary from 2022, i.e., $w(x)=\frac{1}{n}\sum_{i=1}^nx(i,2022))$, is public knowledge. See \cref{fig:structured PP example} for an instance of the described scenario ($n=4$ and $m=2$). 

DP operates by making any pair of neighboring databases indistinguishable, with the definition of neighbors being up to the privacy mechanism designer. In the scenario above, one may apply a DP-based approach by pairing as neighbors every two departments between which there is a difference in the function value $g$ (whether number of employees getting a raise is more than $m$).
For the example from \cref{fig:structured PP example} this amounts to the set of pairs $\{(\mathrm{HR},\mathrm{PR}),(\mathrm{HR},\mathrm{R\&D}),(\mathrm{IT},\mathrm{PR}),(\mathrm{IT},\mathrm{R\&D})\}$ (marked by the dashed purple arrows in the figure). However, by following this approach, we guarantee a stricter privacy requirement than necessary. Indeed, upon observing the privatized version of the published query $f$ and assuming $w$ is publicly known, we only need to make the sets $g^{-1}(0)=\{\mathrm{PR},\mathrm{R\&D}\}$ and $g^{-1}(1)=\{\mathrm{HR},\mathrm{IT}\}$ indistinguishable (marked by the solid dark cyan arrow in \cref{fig:structured PP example}). The benefit of targeting this relaxed notion of privacy is that it enables to achieve improved accuracy and utility. The PP framework is designed to do just that, by enabling full customization of the events that are regarded as private. In addition, PP allows integrating into the framework domain knowledge on the distribution over databases; by considering the set of all possible distributions, this reduces back to the worst-case requirement of DP.

However, the generality of the PP also has drawbacks. For starters, PP does not satisfy general composability \cite{KM14, SoK_DP22}, which is regarded as a privacy axiom---a property that any privacy mechanism should possess. Hence, the outputs of two PP mechanisms can not always be combined to satisfy PP together (as a multi-query output), which limits its usage in practice. In addition, there is a shortage of mechanisms that guarantee PP. The main attempt in that direction is the Wasserstein mechanism from~\cite{SWC17}, which is computationally burdensome as it requires computing the $\infty$-Wasserstein distance between all pairs of conditional distributions of the mechanism's output given any pair of secrets events. Lastly, we note that formal guarantees for PP mechanisms pertaining to privacy-utility tradeoffs, sample complexity bounds for private inference tasks, etc., are largely unavailable due to the hardship of analyzing this framework in full generality. Our goal is to address these shortcomings by introducing some structure into the PP framework to make it more tractable while preserving versatility, and then study the structured variant using tools from information theory.

\subsection{Contributions}

We first propose a novel structured PP framework, where 
the private~and public information is modeled as pairs of functions of the database that are coupled via a bipartite graph. 
This framework captures various privacy notions, from DP \cite{DMNS06} to attribute~privacy (AP) \cite{ZOC20},\footnote{{AP guarantees privacy of functions associated with possibly sensitive attributes in a database, e.g., race or gender.
For instance, referring back to the example in \cref{fig:structured PP example}, if the secret function was the maximum salary of the year 2023 and no public information of the average 2022 salary was available, then that setting would fall under the AP framework. We note that the current example, however, is not captured under AP since the private function is related to two attributes of the database rather than a single attribute as in AP.}}as special cases, while lending well for analysis via tools from information theory. We provide an information-theoretic formulation of the structured PP framework in terms of the conditional mutual information between the mechanism and the secret function given the public one. Generally, the $\epsilon$-mutual information PP ($\epsilon$-MI PP) criteria is implied by $\epsilon$-PP, but we further show that it is sandwiched between $\epsilon$-PP and $(\epsilon,\delta)$-PP in terms of strength under appropriate distributional assumptions and parameter values. The proof relies~on representing PP constraints as bounds on certain  divergences,\footnote{$\infty$-R\'enyi divergence for $\epsilon$-PP and total variation distance for $(0,\delta)$-PP.} and comparing those to the Kullback-Leibler (KL)~divergence (and thus mutual information) via tools like  Pinsker's inequality and the minimax redundancy capacity~theorem. 

{The information-theoretic formulation of the structured PP framework enables a comprehensive analysis of properties, mechanisms, and applications. We begin by establishing properties of $\epsilon$-MI PP mechanisms, encompassing convexity, post-processing, and composability. This shows that our MI PP definition satisfies all the axioms required from a privacy framework \cite{KL12,SoK_DP22}. In particular, while standard PP mechanisms are generally not composable \cite{KM14,SoK_DP22}, our composability results for $\epsilon$-MI PP offer greater flexibility especially in the non-adaptive query setting.}

We next study $\epsilon$-MI PP mechanisms, which is another aspect where the standard PP framework is lacking (the main available mechanisms for standard PP is the Wasserstein mechanism \cite{SWC17}, which is computationally intractable). We derive sufficient conditions on the injected noise level for the Laplace and Gaussian mechanisms that guarantee $\epsilon$-MI PP, and thus also $\epsilon$-MI DP as a special case. The derivation of MI PP mechanisms relies on controlling mutual information via maximum entropy arguments and the entropy power inequality. The resulting noise parameter bounds depend on the conditional variance of the query, which differs from classical results that typically depend on the $\ell^1$- or $\ell^2$-sensitivity of the query; cf. e.g., \cite{DR14,HLM15,BW18,ZWBLRYSLY19}. 
Variance-based parameter bounds are particularly desirable under the PP framework as it encodes prior knowledge on the data distribution. Indeed, it may be the case that sensitivity explodes (e.g., for query functions with unbounded range) but variance is finite due to concentration properties of the distribution class. One drawback of the proposed mechanisms (as well as sensitivity-based ones) is that the injected noise level grows linearly with the dimension. To circumvent this effect, we also propose a Gaussian mechanism that first projects the high-dimensional data onto a low-dimensional space and then adds noise. We obtain parameter bounds for this projection mechanism in terms of the operator norm of conditional covariance matrices and the conditional mean vectors.

Several applications of $\epsilon$-MI PP are explored, starting from auditing for DP \cite{ding2018detecting,jagielski2020auditing,domingo2022auditing}. Auditing black-box mechanisms to certify whether they satisfy a target DP guarantee is challenging, especially in high-dimensional settings. To address this problem, we observe that to audit for DP violations, it suffices to test whether a relaxed privacy notion violates the target privacy level \cite{domingo2022auditing}. We then propose a rigorous hypothesis testing framework for DP violations using the information-theoretic formulation of DP as our test statistic. Since estimating mutual information between high-dimensional variables is statistically burdensome, we introduce a further relaxation to privacy based on sliced mutual information (SMI) \cite{goldfeld2021sliced,goldfeld2022kSMI}, which enjoys parametric empirical convergence rates in arbitrary dimension. Our auditing approach naturally extends to the PP framework. Beyond privacy auditing, we also explore multivariate mean estimation under $\epsilon$-MI PP and derive sample complexity bounds that adapt to the domain knowledge in the PP framework. 
Lastly, we study privacy-utility tradeoffs using $\epsilon$-MI PP and explore its connections to algorithmic stability, which is a standard tool for establishing generalization bounds in statistical learning theory \cite{raginsky2016information,steinke2020reasoning}.

{
\subsection{Related work}
Connections between statistical privacy and information theory have gained increased attention \cite{BK11,DJW13,AAL15,WYZ16,CY16,IW17,mironov2017renyi} as they enable borrowing tools and ideas from one discipline to make progress in the study of the other. In particular, \cite{CY16} established a two-sided connection between DP and the conditional mutual information between the mechanism and any individual record, given the rest of the database. This mutual information-based privacy notion (hereafter abbreviated as MI DP) lends well for an information-theoretic analysis and quantifies privacy via a common currency using which privacy-utility tradeoffs may be explored~\cite{MSSNM14}.
Privacy metrics based on mutual information have been leveraged to analyze and provide guarantees for various inference and learning tasks. MI DP has been used in~\cite{SKT18} to study fundamental privacy-utility tradeoffs in linear regression problems. 
Variants of MI DP have also been used in the context of federated learning study  the generalization error and privacy leakage  \cite{BayesianMIDPFederatedLearning, ImprovedBoundsTogeneralization}, as well as convergence of privacy-preserving training algorithms \cite{MIDPsun2022stochastic}. 
Mutual information-based privacy leakage metrics have also been used in other applications, such as optimal battery charging policies subject to privacy constraints \cite{smartmetering18}, multiple hypothesis testing \cite{guo2022analyzingWithMIDP}, and online location tracing \cite{LocationPrivacyMI8387873}.

Other widely used average-case privacy notion is Kullback-Leibler (KL) DP \cite{KLwang2016average,SoK_DP22} and R\'enyi DP \cite{mironov2017renyi}, both of which serve as relaxations of the classical DP framework. While the main appeal of such average notions is their analytic tractability, they have also been utilized for various applications. KL DP has been applied in settings ranging from collaborative schemes \cite{KLprivacyApp9488663,KLwang2018privacy} and smart grids \cite{KLwang2020disehppc} to the industrial internet of things \cite{KLwang2020ppcs}. It has also been used in tandem with worst-case privacy notions such as DP \cite{KLwang2020disehppc}, by employing them in different stages of the algorithm/scheme of interest. This highlights that even in applications where average-case privacy requirements are not sufficient by themselves, combining them in certain (less sensitive) stages of the system is beneficial, e.g., in terms of utility. KL DP is a special case of R\'enyi DP \cite{mironov2017renyi}  with $\alpha=1$. R\'enyi DP  has been applied to keep track of the privacy budgets in applications including 
optimization \cite{feldman2018privacy}, deep learning \cite{abadi2016deepLearningDP}, and generative adversarial networks \cite{torkzadehmahani2019dp}. The utility and tractability of privacy notions like KL DP, R\'enyi DP, and MI DP serve as inspiration for the information-theoretic formulation of PP proposed herein.
 
 }

\subsection{Organization} 
The rest of the paper is organized as follows. In \cref{Sec: background}, we introduce notation and preliminaries. The structured PP framework, its information-theoretic formulation, and the relation to $\epsilon$- and $(\epsilon,\delta)$-PP are the focus of \cref{Sec:Pufferfish Privacy and Mutual Information}. Properties of the $\epsilon$-MI PP framework and mechanisms are treated in Sections \ref{Sec: properties} and \ref{Sec: Mechanisms}, respectively. In \cref{Sec: Auditing} we design a sample-efficient hypothesis test for privacy auditing based on $\epsilon$-MI PP. Additional applications  to private mean estimation, algorithmic stability, and privacy-utility tradeoffs are covered in \cref{Sec: Additional applications}. Proofs are given in \cref{Sec:proof}, while \cref{Sec: Conclusion} provides concluding remarks and future directions.

\section{Background and Preliminaries} \label{Sec: background}
We set up the notation used throughout this paper, present the DP framework along with its information-theoretic formulation from \cite{CY16}, and introduce the PP paradigm.

\subsection{Notation}

Sets are denoted by calligraphic letters, e.g. $\cX$. For $k,n \in\NN$, we use $\cX^{n \times k}$ for the database space of $n\times k$ matrices (columns correspond to different attributes while rows to different individuals). The $(i,j)$th entry of $x\in\cX^{n\times k}$ is $x(i,j)$. The $i$th row and $j$th column of $x$ are $x(i,\cdot)$ and $x(\cdot,j)$, respectively. The image of a function $g: \cX^{n\times k} \to \RR^d$ is denoted by $\mathrm{Im}(g)$. For $p\geq 1$, $\|\cdot\|_p$ designates the $\ell^p$ norm on $\RR^d$; we omit the subscript when $p=2$ (which is our typical use case). The operator norm for matrices is denoted by $\| \cdot \|_{\op}$. For two numbers $a$ and $b$, we use the notation $a \wedge b = \min \{ a,b \}$ and $a \vee b = \max \{a,b \}$

We denote by $(\Omega,\cF,\PP)$ the underlying probability space on which all random variables (RVs) are defined, with $\EE$ designating expectation. RVs are denoted by upper case letters, e.g., $X$, with $P_X$ representing the corresponding probability law. For $X\sim P_X$, we interchangeably use $\supp(X)$ and $\supp(P_X)$ for the support. The joint law of $(X,Y)$ is denoted by $P_{XY}$, while $P_{Y|X}$ designates the (regular) conditional probability of $Y$ given $X$. 
Conventions for $n\times k$-dimensional random variables are the same as for deterministic elements. The space of all Borel probability measures on $\cS\subseteq\RR^d$ is denoted by $\cP(\cS)$. We write $P\ll Q$ to denote that $P$ is absolutely continuous with respect to (w.r.t.)~$Q$. The $n$-fold product measure of $P\in\cP(\cS)$ is $P^{\otimes n}$. 
Indicator function of a measurable event $\cA \in \cF$ is denoted by $\mathds{1}_\cA$.

For $(X,Y)\sim P_{XY}$, the mutual information between $X$ and $Y$ is denoted by $\sI(X;Y)$. The differential entropy of $X$ is $\sh(X)$. Conditional versions of the above given a third (correlated) RV $Z$ are denoted by $Z$ by $\sI(X;Y|Z)$ and $\sh(X|Z)$, respectively. The Kullback-Leibler (KL) divergence between $P, Q\in\cP(\cX)$ with $P\ll Q$ is 
\[\dkl(P\|Q) :=\EE_P\left[\log\left(\frac{d P}{d Q}\right)\right],\]
where $\frac{d P}{d Q}$ is the Radon-Nikodym derivative of $P$ w.r.t. $Q$. The total variation (TV) distance is defined as
\[\|P-Q\|_{\tv}:=\sup_{\cA}\big|P(\cA) -Q(\cA)\big|,\]
where the supremum is over all measurable sets $\cA$. Both the KL divergence and the TV distance are jointly convex in $(P,Q)$, and are related to one another via Pinsker's inequality \cite{RW09GeneralisedPinsker}: $\|P-Q\|_{\tv}\leq \sqrt{0.5\mathsf{D}_{\mathsf{KL}}(P\|Q)}$. Also recall that $\sI(X;Y)$ can be expressed in terms of KL divergence as $\sI(X;Y)= \mathsf{D}_{\mathsf{KL}}(P_{XY}\|P_X \otimes P_Y),$
where $P_X$ and $P_Y$ are the respective marginals of $X$ and $Y$. For $1\leq p<\infty$, the $p$-Wasserstein distance between $P,Q \in \cP(\cX)$ with finite $p$th absolute moments, i.e., $\EE_P[\|X\|^p],\EE_Q[\|Y\|^p]<\infty$, is $\sW_p(P,Q):=  \inf_{\pi \in \Pi(P,Q)} \big(\EE_\pi\big[\|X-Y\|^p\big]\big)^{1/p}$, where  $\Pi(P,Q)$ is the set of couplings of $P$ and $Q$. The $\infty$-Wasserstein distance is given by $\sW_\infty(P,Q)\coloneqq\lim_{p\to\infty}\sW_p(P,Q)=\inf_{\pi \in \Pi(P,Q)}\sup_{x,y\in\supp(\pi)}\|x-y\|$.

For multi-index $\alpha=(\alpha_1, \ldots, \alpha_d) \in \NN_0$, the partial derivative operator of order $\|\alpha\|_1 $ is $D^\alpha= \frac{\partial^{\alpha_1}}{\partial^{\alpha_1} x_1} \ldots \frac{\partial^{\alpha_d}}{\partial^{\alpha_d} x_d} $. For an open set $\cU \subseteq \RR^d$ and $s\in\NN_0$, let $\sC^s(\cU)$ be the class of functions whose partial derivatives up to order $s$ all exist and are continuous on $\cU$. The H\"older function class of smoothness $s\in\NN_0$ and radius $b\geq 0$ is then defined as $\sC^s_b(\cU):=\{f \in \sC^s(\cU): \max_{\alpha: \|\alpha\|_1 \leq s} \|D^\alpha  f \|_{\infty,\cU} \leq  b  \}$. The restriction of $f: \RR^d \to \RR$ to $\cX \subseteq \RR^d$ is denoted by~$f|_{\cX}$. For compact $\cX$, slightly abusing notation, we set $ \| \cX\| := \sup_{x\in \cX} \|x\|$.

\subsection{Differential Privacy}
DP allows answering queries about aggregate quantities while protecting the individual entries in the database \cite{DMNS06}. To that end, the output of differentially private mechanism should be indistinguishable for neighboring databases---those that differ only in a single record (row). Formally, we say~that $x,x'\in \cX^{n\times k}$ are neighbors, denoted $x\sim x'$,  if $x(i,\cdot)\neq x'(i,\cdot)$ for some $i=1,\ldots,n$, and agree on all other rows.

\begin{definition}[Differential privacy] \label{def: DP}
Fix $\epsilon,\delta>0$. A~randomized mechanism\footnote{A randomized mechanism is described by a (regular) conditional probability distribution given the data, i.e., $P_{M|X}$.} ${M}: \cX^{n \times k} \to \cY$ is $(\epsilon,\delta)$-differentially private if for all $x \sim x'$ with $x,x' \in \cX^{n \times k}$, and $\cA \subseteq \cY $ measurable, we have 
\begin{equation}
\PP\big(M(x) \in \cA \big) \leq e^{\epsilon} \hspace{1mm} \PP\big(M(x') \in \cA\big) + \delta.\label{eq:dp_def}
\end{equation}
\end{definition}

The formulation of DP can be extended to arbitrary neighboring relations between databases, which is known as generic DP \cite{SoK_DP22}. Namely, neighbors can be defined as pairs that agree on all entries except any prespecified portion of the database (as opposed to just the rows, as in standard DP). For instance, viewing databases that agree up to their columns as neighbors, gives rise to a variant of the AP framework \cite{ZOC20}.

\medskip
An information-theoretic formulation of DP was proposed in~\cite{CY16} in terms of the conditional mutual information between each row of the database and the mechanism, given the rest of the rows. We next define $\epsilon$-mutual information DP ($\epsilon$-MI~DP) and then recall the main equivalence result of \cite{CY16}.

\begin{definition}[$\epsilon$-MI DP]\label{def:MI DP}
A Randomized mechanism $M: \cX^{n\times k} \to \cY$ is $\epsilon$-MI DP, if 
\begin{equation}
    \sup_{\substack{P_X \in \cP(\cX^{n \times k}),\\i=1,\ldots,n}} \sI\big(X(i,\cdot);M(X)\big|\big(X(j,\cdot)\big)_{j\neq i}\big) \leq \epsilon.\label{EQ:MI_DP}
\end{equation}
\end{definition}
Theorem 1 of \cite{CY16} states that $\epsilon$-DP (i.e., $(\epsilon,\delta)$-DP with $\delta=0$) implies $\epsilon$-MI DP, which further implies $(\epsilon',\sqrt{2\epsilon})$-DP, for any $\epsilon'\geq 0$. Thus, $\epsilon$-MI DP is in fact sandwiched between $\epsilon$-DP and $(\epsilon,\delta)$-DP in terms of its strength. It was also shown in \cite{CY16} that $\epsilon$-MI DP satisfies properties such as convexity, post-processing, and adaptive/non-adaptive composition.

\subsection{Pufferfish Privacy}
For a database space $\cX^{n \times k}$, the PP framework~\cite{KM14} consists of three components: (i) a set of secrets $\mathcal{S}$, that contains measurable subsets of $\cX^{n \times k}$; (ii) a set of secret pairs $\mathcal{Q} \subseteq \mathcal{S} \times \mathcal{S}$ that needs to be statistically indistinguishable in the $(\epsilon,\delta)$ sense (see~\eqref{eq:pp_def}); and (iii) a class of data distributions $\Theta \subseteq \cP(\cX^{n \times k})$, that captures prior beliefs or domain knowledge. As formulated next, the goal of PP is to make all secret pairs in $\cQ$  indistinguishable w.r.t. those prior beliefs $P_X \in \Theta$.  
\begin{definition}[Pufferfish privacy]\label{def:pp_def}
Fix $\epsilon,\delta>0$. A~randomized mechanism ${M}: \cX^{n \times k} \to \cY$ is $(\epsilon,\delta)$-private in the pufferfish framework $(\mathcal{S},\mathcal{Q}, \Theta )$ if for all $P_X \in \Theta$, $(\cR,\cT) \in \mathcal{Q}$ with $P_X(\cR),P_X(\cT)> 0$, and $\cA \subseteq \cY $ measurable, we have
\begin{equation}
\PP\big(M(X) \in \cA \big| \cR\big) \leq e^{\epsilon} \hspace{1mm} \PP\big(M(X) \in \cA\big| \cT\big) + \delta.\label{eq:pp_def}
\end{equation}
\end{definition}

DP from Definition~\ref{def: DP} is a special case of PP when $\cS=\cX^{n\times k}$, $\cQ$ contains all neighboring pairs of databases, and $\Theta=\cP(\cX^{n \times k})$ (i.e., no distributional assumptions are made, and privacy is guaranteed in the worst case). PP also subsumes any other famework under generic DP \cite{SoK_DP22} (i.e., alternative neighboring relations) by choosing $\cQ$ accordingly. 
Another special case of PP is AP \cite{ZOC20}, which privatizes global statistical properties of data attributes. In this case, $\cS$ is the value of a function evaluated on the data, $\mathcal{Q}$ contains pairs of function values, and $\Theta$ captures assumptions on how the data was sampled and correlations across attributes thereof. These special cases are discussed in detail in Remark \ref{rem:pp_reductions} ahead.

\section{Pufferfish Privacy and Mutual Information} \label{Sec:Pufferfish Privacy and Mutual Information}

Towards an information-theoretic characterization of PP, it is convenient to focus on a slightly more structured formulation that explicitly decomposes pairs of secrets into private and public parts. The considered PP framework is presented next, followed by an information-theoretic characterization.

\subsection{Structured Pufferfish Privacy Framework}

We focus on a special case of the general framework, where pairs $(\cR,\cT)\in\cQ$ are decomposed into a private part (on which they should be indistinguishable) and a common part (interpreted as public information). In Remark~\ref{rem:pp_reductions} we demonstrate how the considered formulation reduces to popular privacy notions like DP~\cite{DMNS06} and AP~\cite{ZOC20}. Our formulation is constructed as follows:

\begin{enumerate}[wide, labelindent=10pt]
    \item \underline{Private/public functions:} Let $\cG$ and $\cW$ be finite sets, 
    containing functions on $\cX^{n\times k}$. For $g\in\cG$, we interpret $g(X)$ as a private feature of the database $X\sim P_X\in \Theta$, while $w(X)$, $w\in\cW$, represents publicly available information. 
    \item \underline{Function pairs:} To encode which private-public function pairs constitute a secret (i.e., an element of~$\cS$) we use a bipartite graph. Consider the graph $(\cG,\cW,\cE)$, where $\cE$ is a given edge set between the two partitions $\cG$ and $\cW$. We write $g\sim w$ if $\{g,w\}\in\cE$ for some $g\in\cG$ and $w\in\cW$. The operational meaning of an edge $g\sim w$ is that $g(X)$ must be concealed even if the adversary has access to $w(X)$. For example, for DP we take $g_i(x)=x(i,\cdot)$ as a specific row of the database and $w_i(x)=\big(x(j,\cdot)\big)_{j\neq i}$ as the rest of the database, where $i=1,\ldots,n$; then set $\cG=\{g_i\}_{i=1}^n$, $\cW=\{w_i\}_{i=1}^n$, and $\cE=\big\{\{g_i,w_i\}\big\}_{i=1}^n$, as depicted in Figure~\ref{fig:DP_function_pairs}.)

\begin{figure}[t!]
\begin{center}\hspace{13mm}
\begin{psfrags}
\psfragscanon
\psfrag{A}[][][0.9]{\hspace{-1mm}$n$}
\psfrag{B}[][][0.9]{$k$}
\psfrag{C}[][][0.9]{\color{brown}\hspace{-20mm}$g_i(x)=x(i,\cdot)$}
\psfrag{D}[][][0.9]{ \hspace{-30mm}$w_i(x)=\big(x(j,\cdot)\big)_{j\neq i}$ }
\psfrag{E}[][][0.9]{\color{brown}\hspace{-22mm}(private portion)}
\psfrag{F}[][][0.9]{\hspace{-25mm}(public portion)}
\psfrag{G}[][][0.9]{$x \in \cX^{n\times k}$}
\includegraphics[scale = .7]{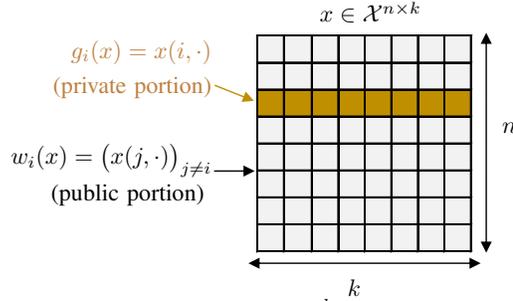}
\vspace{-5mm}
\caption{Function pairs for DP: The $i$th row of $x \in \cX^{n \times k}$ is the private portion, while the rest of the database is the corresponding public part. 
\vspace{-2mm}} 
\label{fig:DP_function_pairs}
\psfragscanoff
\end{psfrags}
\end{center}
\end{figure}

    \item \underline{Secret event:} Each secret event (namely, an element of $\cS$) corresponds to a specific value that a private-public function pair takes, i.e., for $\cG\ni g\sim w\in\cW$, $a\in\mathrm{Im}(g)$, and $c\in\mathrm{Im}(w)$, $\cS$ comprises all events of the form $\cA_{g,w}(a,c):=\big\{g(X)=a,\,w(X)=c\big\}$. 
    \item \underline{Secret event pairs:} Elements of $\cQ\subseteq \cS\times\cS$ are pairs that share the same public information (i.e., the value for $w(X)$) but differ in their private portions (the value of $g(X)$). 
\end{enumerate}

We are now ready to define the structured PP framework. 

\begin{definition}[Structured PP framework]\label{def:specialized_PP}
Fix $\epsilon,\delta>0$ and consider a bipartite graph $(\cG,\cW,\cE)$ with sets of functions $\cG$ and $\cW$ as described above. A randomized mechanism $M:\cX^{n\times k}\to\cY$ is $(\epsilon,\delta)$-private in the structured pufferfish framework $(\cG,\cW,\cE,\Theta)$ if it satisfies Definition~\ref{def:pp_def} with
\begin{align*}
 \mathcal{S}\mspace{-2mu} &=\mspace{-2mu}    \big\{\cA_{g,w}(a,c):\cG\ni g\sim w\in\cW ,\,a \in \mathrm{Im}(g),\, c \in \mathrm{Im}(w)\big\} 
 \cr
 \mathcal{Q}\mspace{-2mu}&=\mspace{-2mu} \big\{\{\cA_{g,w}(a,c) , \cA_{g,w}(b,c) \big\}\mspace{-2mu}: \cG\ni g\mspace{-2mu}\sim\mspace{-2mu} w\in\cW,\,c \in \mathrm{Im}(w),
\ a,b \in \mathrm{Im}(g),\,a \neq b \big\}
\end{align*}
and a set of data distributions $\Theta\subseteq\cP(\cX^{n\times k})$.
\end{definition}

The structured PP framework captures various prominent privacy notions, such as DP \cite{DMNS06} and AP \cite{ZOC20}.

{ 

\begin{remark}[Semantics of the structured PP framework] \label{rem:semantics of structured PP}
Structured PP framework provides the following privacy guarantee: for any database $X$ generated from a distribution in the class $\Theta$, an adversary that knows the function value $w(X)$, for $w\in \cW$, and the output of the privatization mechanism $M(X)$ draws the same conclusions regardless of the value of the private function $g(X)$, $g \in \cG$. This applies to many realistic scenarios, such as the one described in  \cref{subsec:PP_intro} (see also \cref{fig:structured PP example}). Noting that this example indeed falls under the structured PP framework, we have that the value of $g$ (whether more than half of the employees in each department received a raise) is protected even if the adversary has the access to $w$ (average salaries of year 2022 across the department) and $M$ (the privatized average 2023 salary).
\end{remark}
}

\begin{remark}[Special cases] \label{rem:pp_reductions}
The structured PP framework reduces to various important privacy notions. We provide two such examples pertaining to DP and AP.
\begin{enumerate}
    \item DP corresponds to a structured PP framework with $\Theta= \cP(\cX^{n \times k})$, private functions $g_i(x)=x(i,\cdot)$, public functions $w_i(x)=\big(x(j,\cdot)\big)_{j\neq i}$, where $i=1,\ldots,n$, and an edge set $\cE=\big \{ \{g_i,w_i\} \big\}_{i=1}^n$. This construction naturally extends to any privacy framework where secret events are singletons (databases). Then, each private function $g$ acts on some prespecified portion of the database, and is connected by an edge to a public function $w$ that acts on the remaining data entries. 
    \item The AP framework privatizes attributes of the database, which are captured by certain functions $\tilde{g}_j:\cX^k\to\RR$ of the columns $j=1,\ldots,k$. Following the setup of \cite{DMNS06}, we take $g_j(x)=\tilde{g}_j\big(x(\cdot,j)\big)$, for $j=1,\ldots,k$; as AP includes no public information we set $\cW=\cE=\emptyset$ and let $\Theta$ be the class of distributions of interest. Alternatively, one may consider a variant of AP with public information, which is the portion of the database except the considered column. In that case, $\cW$ is a set of functions $\smash{w_j(x)=\big(x(\cdot,i)\big)_{i\neq j}}$, where $j=1,\ldots,k$, and $\cE=\big \{ \{g_j,w_j\} \big\}_{j=1}^k$. 
\end{enumerate}
\end{remark}

{Given the definition of the structured PP framework, it is natural to ask for mechanisms that attain it. The Wasserstein mechanism from \cite{SWC17} can be used to guarantee general PP. However, 
that approach is computationally intractable. A simpler mechanism can be devised by following the approach of \cite[Theorem~1]{ZOC20} for AP under appropriate Gaussianity assumptions. Specifically, fix a query $f: \cX^{n \times k} \to \RR$ and suppose that for any $X\sim P_X\in\Theta$, $g\in\cG$, and $w\in\cW$ with $g\sim w$, we have that the conditional distribution of $f(X)$ given $\big(g(X),w(X)\big)$ is Gaussian. Under this assumption, the conditional variance $\mathrm{Var}\big(f(X) \big | g(X)=a, w(X)=c \big)$ does not depend on the values $(a,b)$ and the Gaussian noise-injection mechanism $M(X)=f(X)+Z$, with $Z \sim \cN(0,\sigma^2)$ satisfies $(\epsilon,\delta)$-structured PP whenever
\[\sigma^2 \geq \sup_{\substack{P_X \in \Theta,\\g\in\cG,w\in\cW:\\g \sim w, \\ c \in \mathrm{Im}(w)}} 2\left(\epsilon^{-1}\Delta^{P_X}_{f,g,w}(c)\right)^2\log(1.25/\delta)  -\mathrm{Var}\big(f(X) \big | g(X)=a_0, w(X)=c_0 \big), \]
where $(a_0,c_0)\in\mathrm{Im}(g)\times \mathrm{Im}(w)$ are arbitrary and 
\[\Delta^{P_X}_{f,g,w}(c) \coloneqq \sup_{a,b \in \mathrm{Im}(g)} \big| \EE\big[ f(X) \big | g(X)=a, w(X)=c  \big] -  \EE\big[ f(X) \big | g(X)=b, w(X)=c  \big] \big |.\]
While the above Gaussianity requirement may hold by assuming, e.g., Gaussian data and linear functions, it is generally quite restrictive. To devise tractable mechanisms beyond this Gaussian setting, we next propose an information-theoretic reformulation of the structured PP framework. This reformulation lends well for analysis and enables deriving sufficient conditions on parameters of noise-injection mechanism that guarantee structured PP in general.

}

\subsection{Information-Theoretic Formulation}
To provide an information-theoretic formulation of the structured PP framework, we first define $\epsilon$-MI PP. 

\begin{definition}[$\epsilon$-MI PP]\label{def:MI_PP_def}
Let $(\cG,\cW,\cE)$ be a bipartite graph as in Definition ~\ref{def:specialized_PP} and $\Theta\subseteq\cP(\cX^{n\times k})$. A randomized mechanism $M:\cX^{n\times k}\to\cY$ is $\epsilon$-MI PP in the framework $(\cG,\cW,\cE,\Theta)$ if
\[
\sup_{\substack{P_X \in \Theta,\\g\in\cG,w\in\cW:\\g \sim w}} \sI\big(g(X);M(X)|w(X)\big) \leq \epsilon.
\]
\end{definition}
Evidently, $\epsilon$-MI PP as defined above recovers the notion of $\epsilon$-MI DP from \cite{CY16} (see Definition~\ref{def:MI DP}) by taking $(\cG,\cW,\cE,\Theta)$ as described in Part 1 of Remark ~\ref{rem:pp_reductions}.

{
\begin{remark}[Revisiting semantics of the structured PP]\label{rem:semantics_MI_PP}
The $\epsilon$-MI PP formulation explicitly encodes the semantics of the structured PP framework, as explained in \cref{rem:semantics of structured PP}. Namely, for any database distribution $P_X \in \Theta,\cG \ni g \sim w \in \cW$, the mechanism's output $M(X)$ should not convey more than $\epsilon$ information bits about any secret function $g(X)$, $g\in\cG$, even when $w(X)$, $w\in\cW$, is available as side information. 
\end{remark}

}

The following theorem characterizes the relative strength of the structured PP framework from Definition~\ref{def:specialized_PP} compared to $\epsilon$-MI PP, showing that the latter lies between $\epsilon$-PP (i.e., $(\epsilon,\delta)$-PP with $\delta=0$) and $(\epsilon,\delta)$-PP for appropriate parameter values. 

\begin{theorem}[Relative strength]\label{thm:equivalence}
Consider the structured $(\epsilon,\delta)$-PP framework $(\cG,\cW,\cE,\Theta)$ from Definition~\ref{def:specialized_PP}. Let $\epsilon'>0$ be arbitrary and set $\epsilon''=\epsilon\wedge \frac 12 \epsilon^2$. Then 
\[
\epsilon\textnormal{-PP} \implies  \epsilon'' \textnormal{-MI PP }. 
\]
and if $\Theta=\cP(\cX^{n \times k})$, then we further have
\[
\epsilon\textnormal{-PP} \implies  \epsilon'' \textnormal{-MI PP } \implies (\epsilon', \sqrt{2\epsilon''})\textnormal{-PP}.
\]

Moreover, the inverse implication 
\[
  (\epsilon,\delta)\textnormal{-PP} \implies \epsilon^{\star} \textnormal{-MI PP }  
\]
holds under either of the following conditions:
\begin{enumerate}
    \item $\big|\supp\big(M(X)\big)\big|<\infty$ or $\max_{g\in\cG}| \mathrm{Im}(g)|<\infty$, whence \[\epsilon^{\star}=2\sh_b(\delta^{'}) +2 \delta^{'} \log \left(\big|\supp\big(M(X)\big)\big|\wedge \big(\max_{g \in \cG}|\mathrm{Im}(g)|+1\big)\right)\] where $\sh_b(\alpha)=-\alpha\log(\alpha)-(1-\alpha)\log(1-\alpha)$, for $\alpha\in[0,1]$, is the binary entropy function in nats and $\delta^{'} =1-2(1-\delta)/(e^\epsilon+1)\in [0,1]$.
    \item The joint density $f_{M(X),g(X),w(X)}$ and conditional density $f_{M(X)|g(X),w(X)}$  exists, whence 
     \[\epsilon^\star = 
    \left(1- \frac{2(1-\delta)}{e^\epsilon+1} \right) \left\{ {\sup_{\substack{P_X \in \Theta,\\(g,w)\in\cG\times\cW:\, g \sim w,\\ a,b \in \mathrm{Im}(g),\, c \in \mathrm{Im}(w)}}} \frac{1}{2}\left( \frac{\log \big(\alpha_{a,b,c}^{-1}\big)}{1-\alpha_{a,b,c}} - \beta_{a,b,c} \right) \wedge  {\sup_{\substack{P_X \in \Theta,\\(g,w)\in\cG\times\cW:\, g \sim w,\\ a \in \mathrm{Im}(g),\, c \in \mathrm{Im}(w)}}}\log \left( \frac{u_{a,c}}{\ell_{a,c}} \right) \right\} ,\]
    where $\alpha_{a,b,c}^{-1}=\sup_{x\in \supp(M(X))}\frac{f_{M(X)|g(X),w(X)}(x|a,c)}{f_{M(X)|g(X),w(X)}(x|b,c)}$, $\beta_{a,b,c}=\inf_{x\in \supp({M(X)})} \frac{f_{M(X)|g(X),w(X)}(x|a,c)}{f_{M(X)|g(X),w(X)}(x|b,c)}$, $u_{a,,c}=\sup_{x\in \supp (M(X))}f_{M(X),g(X),w(X)}(x,a,c)$ and $\ell_{a,c}=\inf_{x\in \supp(M(X))} f_{M(X)|g(X),w(X)}(x,a,c)$. 
   
\end{enumerate}
\end{theorem}

Theorem~\ref{thm:equivalence} is proven in Section ~\ref{proof:equivalence}. The first implication follows by reformulating $\epsilon$-PP in terms of the $\infty$-R\'enyi divergence, translating that to an $\epsilon$ bound on the corresponding KL divergence, and then use joint convexity to arrive at $\epsilon''$-MI PP. When $\Theta$ is the set of all database distributions, the second implication is derived via the minimax redundancy capacity representation and Pinsker's inequality. The inverse implications first translates $(\epsilon,\delta)$-PP into a bound on the TV distance between corresponding conditional distributions and then employs either continuity of entropy w.r.t. the TV distance to control mutual information or the reverse Pinsker inequality.{Note that the privacy guarantee provided by $\epsilon$-PP is stronger than that of $\epsilon$-MI PP, as evident from the implication $\epsilon$-PP $\implies \epsilon$-MI PP.
}

{
\begin{remark}[$\epsilon$-KL PP]\label{rem:KL_PP} $\epsilon$-KL DP \cite{KLwang2016average,SoK_DP22} can be generalized to the setting of structured PP as follows. Let $(\cG,\cW,\cE)$ be a bipartite graph as in Definition ~\ref{def:specialized_PP} and $\Theta\subseteq\cP(\cX^{n\times k})$. A randomized mechanism $M: \cX^{n \times k} \to \cY$ is $\epsilon$-KL PP in the framework $(\cG,\cW,\cE,\Theta)$ if 
\[
 \dkl\big( P_{M(X)|\cA_{g,w}(a,c)} \big \| P_{M(X)|\cA_{g,w}(b,c)} \big) \leq \epsilon, \quad \forall P_X \in \Theta,\, \cG \ni g \sim w \in \cW,\, a,b \in \mathrm{Im}(g),\, c \in \mathrm{Im}(w),
\]
where $\cA_{g,w}(a,c)=\big\{g(X)=a,\,w(X)=c\big\}$. $\epsilon$-KL PP as defined above sits between the structured PP from \cref{def:specialized_PP} and $\epsilon$-MI PP from \cref{def:MI_PP_def} in terms of strength. This is evident from the proof of the first implication in \cref{thm:equivalence}, which effectively shows that $\epsilon$-PP $\implies \epsilon''$-KL PP $\implies \epsilon''$-MI PP, for $\epsilon''=\epsilon\wedge \frac 12 \epsilon^2$ (see \eqref{EQ:KL_MI} in Section ~\ref{proof:equivalence}). The above observation also applies for an extension of $\alpha$-R\'enyi DP \cite{mironov2017renyi} to the structured PP setting. 
\end{remark}
 
 }
\section{Properties of $\epsilon$-MI PP } \label{Sec: properties}

We now explore the properties of $\epsilon$-MI PP, encompassing convexity, post-processing, and composability (adaptive and non-adaptive). Modern guidelines for privacy frameworks \cite{KL12} pose properties such as convexity and post-processing (also known as transformation invariance) as base requirements. Composability is another important property that implies that the joint distribution of the outputs of (possibly adaptively chosen) privacy mechanisms is in itself private. These properties are shown to hold for the general $\epsilon$-PP framework in \cite{KM14}. The next theorem shows $\epsilon$-MI PP satisfies them as well.

\begin{theorem}[Properties of $\epsilon$-MI PP mechanisms] \label{thm:MIPP_properties}
The following properties hold: 
\begin{enumerate}[wide, labelindent=10pt]
     \item \underline{Convexity}:  Let $\epsilon>0$, and $M_1,\ldots,M_k$ be $\epsilon$-MI PP mechanisms. Take $I$ as a $k$-ary categorical random variable with parameters $(p_1,\ldots,p_k)$. Then the mechanism $M:=M_I$ (i.e., $M=M_i$ with probability $p_i$, for $i=1,\ldots,k$) also satisfies $\epsilon$-MI PP.
     \item \underline{Post-processing}: If mechanism $M:\cX^{n\times k}\to\cY$ satisfies $\epsilon$-MI PP, then for any randomized function $A:\cY\to\cZ$, the processed mechanism $A \circ M$ also satisfies $\epsilon$-MI PP.
     \item \underline{Adaptive composability}: 
    Let $M_1,\ldots,M_k$ be sequentially and adaptively chosen $\epsilon_1\ldots,\epsilon_k$-MI PP mechanisms, i.e., 
    \[ \sup_{\substack{P_X \in \Theta,\\g\in\cG,w\in\cW:\\g \sim w}} \mathsf{I}\big(g(X);M_i(X)\big|w(X),M_1(X),...,M_{i-1}(X)\big) \leq \epsilon_i,\quad \forall i=1,\ldots,k.\] 
    Then the composition $M^k=(M_1,\ldots,M_k)$ satisfies $\big(\sum_{i=1}^k \epsilon_i\big)$-MI PP.
\end{enumerate}
\end{theorem}

Theorem~\ref{thm:MIPP_properties} is proven in Section~\ref{proof:MIPP_properties} using basic properties of mutual information, such as the chain rule, the data processing inequality, and its nullification under independence. The simplicity of the argument highlights the virtue of the information-theoretic formulation of the PP framework.

\medskip

We move to discuss non-adaptive composition. In this case, the mechanisms $M_1,\ldots,M_k$ from property~(3) of Theorem \ref{thm:MIPP_properties} are chosen conditionally independent given the database. This instance is of practical importance since it includes noise injection mechanisms (e.g., Gaussian and Laplace), that are the focus on the next section. The following proposition is proven in Section~\ref{Proof:thm:general_nonAdaptivityMIPP}.

\begin{proposition}[Non-adaptive composability]\label{thm:general_nonAdaptivityMIPP}
Let $M_1,\ldots,M_k$ be MI PP mechanisms with the parameters $\epsilon_1,\ldots,\epsilon_k$, respectively, which are chosen non-adaptively, i.e., $P_{M_1,\ldots,M_K|X}=\prod_{i=1}^kP_{M_i|X}$. Then the composition $M^k$ is $\big(\sum_{i=1}^k \epsilon_i+ \eta\big)$-MI PP, where 
\[\eta = \sup_{\substack{P_X \in \Theta,\\g\in\cG,w\in\cW:\\g \sim w}}\sum_{i=2}^k \sI\big(M_i(X);M^{i-1}(X)|w(X),g(X)\big),\]
and $M^{i-1}(X)=\big(M_1(X),\ldots,M_{i-1}(X)\big)$.
\end{proposition}

Proof of Proposition~\ref{thm:general_nonAdaptivityMIPP} (in Section~\ref{Proof:thm:general_nonAdaptivityMIPP}) follows from the repetitive application of chain rule for mutual information and by the fact, entropy being reduced by conditioning.

\begin{remark}[Bounds on $\eta$]
If $|\supp\big(M_i(X))| < \infty$ for each $i=2,\ldots,k$, then $\eta \leq \sum_{i=2}^k\log \left| \supp\big(M_i(X)\big)\right|$. Thus, if the cardinality of the output of the mechanism is small, then so is $\eta$.
Alternatively, if each mechanism conditioned on the database $X$ is log-concave (this is satisfied by the Laplace and Gaussian noise injection mechanisms introduced in Section~\ref{Sec: Mechanisms}) and its output is one-dimensional, then 
\[ \eta \leq \sup_{\substack{P_X \in \Theta,\\g\in\cG,w\in\cW:\\g \sim w}} \frac{1}{2}\sum_{i=2}^k  \EE \left[\log \left( \frac{ \pi e \mathrm{Var}\big(M_i(X)\big|g(X),w(X)\big) }{4 \mathrm{Var}\big(M_i(X)\big|X\big)}\right) \right].\]
This follows from the Gaussian distribution achieving maximum entropy under a variance constraint, and the lower bound for the entropy of log-concave distributions \cite{marsiglietti2018lower}. 
\end{remark}

\begin{remark}[Composition when secret pairs are databases]
It was shown in \cite{KM14} that standard $\epsilon$-PP mechanisms compose in PP frameworks in which secret pairs $(\cR,\cT)\in\cQ$ correspond to pairs of databases (i.e., when $\cS$ contains only singletons; see Definition \ref{def:pp_def}). This also holds for $\epsilon$-MI PP mechanisms by observing that in this case we have $\eta=0$ in Proposition~\ref{thm:general_nonAdaptivityMIPP}. Indeed, as $\big(g(X),w(X)\big)$ specify a database, the conditional independence of the mechanism given $X$ nullifies the mutual information. 
\end{remark}

The general non-adaptive setting, without assuming that secrets specify databases, was studied in \cite{KM14}, where it was shown that composability does not hold in general. \cite{KM14} then identified a (rather restrictive) sufficient condition on the class of distributions $\Theta$, termed \textit{universally composable} (UC) distributions, under which non-adaptive composability holds for $\epsilon$-PP. The class of UC distributions is defined next.

 \begin{definition}[UC distributions]\label{def:UC-distributions}
The class $\Theta_{\mathsf{UC}}$ of UC distributions contains all $P_X\in\cP(\cX^{n \times k})$, such that for all $\cG\ni g\sim w\in\cW $ and $(a,c)\in \mathrm{Im}(g)\times\mathrm{Im}(w)$ with $P_X\big(A_{g,w}(a,c)\big)>0$, we have $P_{X|A_{g,w}(a,c)}=\delta_x$, for some $x\in\cX^{n \times k}$, where $\delta_x$ is the Dirac measure at $x$.
\end{definition}
In words, UC distributions are ones under which the database is specified by non-null secret events.

\medskip

$\epsilon$-MI PP also composes under the UC condition, but turns out to be more stable than the standard PP framework w.r.t. addition on non-UC distributions to $\Theta$. The next corollary, which follows directly from Proposition~\ref{prop:Universal composability of MI PP}, quantifies this fact.

\begin{corollary}[Universal composability]\label{prop:Universal composability of MI PP}
Let $M_1,\ldots,M_k$ be mutual information PP mechanisms with the parameters $\epsilon_1,\ldots,\epsilon_k$, respectively, which are chosen non-adaptively. Then the composition $M^k$ is $\big(\sum_{i=1}^k \epsilon_i\big)$-MI PP, provided either of the following conditions holds:
\begin{enumerate}
    \item [(i)] $\Theta \subseteq \Theta_{\mathsf{UC}}$; or 
    \item[(ii)] $\Theta_{\mathsf{UC}}\subseteq\Theta$ and $M_1,\ldots,M_k$ satisfy standard PP with the same $\epsilon_1,\ldots,\epsilon_k$ parameters. 
\end{enumerate}
\end{corollary}

{The proof of \cref{prop:Universal composability of MI PP} is given in \cref{proof: Universal composability of MI PP}. }
Notably, Case (ii) above states that non-adaptive composability of $\epsilon$-PP mechanisms holds in the sense of $\epsilon$-MI PP whenever $\Theta$ contains all UC distributions. This means that $\epsilon$-MI PP non-adaptive composability is stable to addition of non-UC distributions to $\Theta$, so long that all UC distributions are there (e.g., when $\Theta =\cP(\cX^{n \times k})$). Standard $\epsilon$-PP does not share this stability: even if $\Theta$ contains all UC distribution, adding even a single non-UC distribution to this set will compromise the composability of the classic PP framework.

\section{Mechanisms} \label{Sec: Mechanisms}

This section leverages the information-theoretic formulation to devise Laplace and Gaussian noise-injection $\epsilon$-MI PP mechanisms whose noise level is specified in terms of elementary quantities. {As a special case of MI PP, we obtain mechanisms for MI DP---a framework proposed and studied in \cite{CY16}, but mechanisms were not considered in that work. Mechanisms achieving MI DP tailored for specific applications, such as linear regression and coded federated learning, were developed in \cite{SKT18,BayesianMIDPFederatedLearning}. In contrast, this sequel provides mechanisms that are applicable in general, under minimal assumptions on the setting.}

\subsection{Laplace Mechanism}

Given a query function $f: \mathcal{X}^{n\times k} \to \RR^d $ and a database $X\sim P_X\in\Theta$, a noise-injection mechanism for privately publishing $f(X)$ outputs $M(X)=f(X)+Z$, where $Z$ is a noise variable that follows a prescribe distribution with appropriately turned parameters. The following theorem characterizes parameter values for the Laplace mechanism (i.e., when $Z$ follows the Laplace distribution) that guarantee $\epsilon$-MI~PP.

\begin{theorem}[Laplace mechanism] \label{thm:LaplaceMechanismMIPP}
Fix $\epsilon>0$ and a structured PP framework $(\cG,\cW,\cE,\Theta)$. Let $f: \mathcal{X}^{n\times k} \to \RR^d $ be the query for privatization and consider the Laplace mechanism $M_\mathsf{L}(X):=f(X)+Z_\mathsf{L}$, where $Z_\mathsf{L} \sim \mathrm{Lap}(0,b)^{\otimes d}$ is a $d$-dimensional {{product}} Laplace distribution with the scale parameter $b>0$. If
\[b \geq \sup_{P_X \in \Theta,\, w \in \cW^\star} \frac{\sum_{j=1}^d \EE \sqparen{\sqrt{\mathrm{Var}\big(f_j(X)|w(X)\big)}}}{d(e^{\frac{\epsilon}{d}}-1)},\]
where $f_j(X)$ is the $j$th entry of $f(X)=\big(f_1(X),\ldots,f_d(X)\big)$ and $\cW^\star=\{w \in \cW:\,  \exists\, g \in \cG,\,g \sim w \}$, then $M_\mathsf{L}$ is $\epsilon$-MI PP. 
\end{theorem}
The derivation of Theorem~\ref{thm:LaplaceMechanismMIPP} is presented in Section \ref{Proof: Laplace mechanism MIPP} and relies on the fact that the Laplace distribution maximizes differential entropy subject to an expected absolute deviation constraint. 

\begin{remark}[Comparison with Wasserstein mechanism]
Compared to computing $\infty$-Wasserstein distances for each secret pair for each distribution, variance is an elementary quantity that can be computed with relative ease. There are also scenarios where the noise level induced by Wasserstein mechanism is infinite and thus infeasible, while our Laplace mechanism derives a feasible, finite noise level. For instance, consider the setup of AP~\cite{ZOC20}: if $\Theta$ contains a distribution with respect to which $f(X)$ and $g(X)$ are jointly Gaussian for some $g\in \cG$, variance is finite while $\infty$-Wasserstein distance may diverge \footnote{Let the said joint distribution $P_{f(X),g(X)}$ be $\cN( \mu, \Sigma)$ with $\mu=(\mu_f, \mu_g)^\mathsf{T}$ and $\Sigma=[ (\sigma^2_f, \rho \sigma_f \sigma_g);(\rho \sigma_f \sigma_g, \sigma^2_g)]$, Then, $W_2^2\big(P_{f(X)|g(X)=a},P_{f(X)|g(X)=b}) \big) =|a-b|^2 \sigma^2_f \rho^2 / \sigma^2_g-(1-\rho^2)\sigma^2_f. $ When supremized over $(a,b) \in \mathrm{Im}(g)$, 2-Wasserstein distance diverges. Due to the monotonicity of Wasserstein distance, indeed $\infty$-Wasserstein distance explodes.}.

\end{remark}

The next corollary specializes Theorem~\ref{thm:LaplaceMechanismMIPP} to $\epsilon$-MI DP (see Part 1 of Remark~\ref{rem:pp_reductions}) by controlling the conditional variance in terms of $\ell^1$-sensitivity. The proof is deferred to Section~\ref{Sec: Proof of Cor LaplaceMIDP}.

\begin{corollary}[Laplace mechanism for $\epsilon$-MI DP]\label{cor:LaplaceMIDP}
Under the setup of Theorem~\ref{thm:LaplaceMechanismMIPP}, Laplace mechanism with
\[b \geq 
\sup_{\substack{P_X \in \cP(\cX^{n \times k}),\\i=1,\ldots,n}}\frac{\sum_{\ell=1}^d \EE \sqparen{\sqrt{\mathrm{Var}\big(f_\ell(X)|\big(X(j,\cdot)\big)_{j\neq i}\big)}}}{d(e^{\frac{\epsilon}{d}}-1)},\] is $\epsilon$-MI DP. Furthermore, the the statement remains true if the right-hand-side (RHS) above is replaced with $\frac{\Delta_1(f)}{\sqrt{2}d(e^{\frac{\epsilon}{d}}-1)}$, where  $\Delta_1(f):=\max_{x\sim x'} \| f(x) -f(x')\|_1 $.

\end{corollary}

\begin{remark}[Classic Laplace mechanisms for DP]
Classical Laplace mechanisms achieve $\epsilon$-DP~when $b~\geq~\Delta_1(f) / \epsilon$~\cite{DR14}, and $(\epsilon,\delta)$-DP when $b \geq~ \Delta_1(f)/\big(\epsilon -~ \log(1-\delta)\big)$~\cite{HLM15}.  Evidently, for small $\epsilon$ (termed the `high privacy regime'), both the classic and the $\epsilon$-MI DP mechanism from Corollary~\ref{cor:LaplaceMIDP} induce noise of order $O(1/\epsilon)$. 

\end{remark}

\begin{remark}[Domain knowledge for DP]\label{rem:LaplaceMIDP_Domain}\
Compared to the classic sensitivity-based Laplace mechanisms for DP, the bound in Corollary \ref{cor:LaplaceMIDP} depends on the variance of $f$ and allows to incorporate domain knowledge. 
Consider the product Gaussian family
    \[\Theta_\mathsf{G}(m,s)=\Big\{\prod_{i=1}^n \theta_i: \, \theta_i=\cN(\mu_i,\sigma_i^2), \ |\mu_i| \leq m, \ \sigma_i^2\leq s \Big\},\] 
and let $f(X)=n^{-1}\sum_{i=1}^nX_i$ be the average of the database entries (the argument holds for any linear query). The noise derived from our mechanism is $\sqrt{s}/\big(n(e^{\epsilon}-1)\big)<\infty$, while $\Delta_1(f)=\infty$ here since $X$ has unbounded support. Thus, the sensitivity-based mechanisms are vacuous for this case, while our bound provides feasible noise levels. In these situations, the classic approach involves truncating the space~\cite{kamath2020private} which is not necessary under our framework, whenever the variance is finite. 
In Section~\ref{sec: Private mean estimation} we also discusses the benefits of domain knowledge for private mean estimation tasks. 
\end{remark}

\subsection{Gaussian Mechanism}

We next characterize parameter values for the Gaussian $\epsilon$-MI~PP mechanism. 

\begin{theorem}[Gaussian mechanism] \label{thm:GaussianMechanismMIPP}
Fix $\epsilon>0$ and a structured PP framework $(\cG,\cW,\cE,\Theta)$. Let $f: \mathcal{X}^{n\times k} \to \RR^d $ and consider the Gaussian mechanism $M_\mathsf{G}(X):=f(X)+Z_\mathsf{G}$, where $Z_\mathsf{G} \sim \cN(0,\sigma^2 \mathrm{I}_d)$ is a $d$-dimensional {isotropic} Gaussian of parameter $\sigma>0$. If
\[\sigma^2 \geq \sup_{P_X \in \Theta,\, w \in \cW^\star} \frac{ \sum_{j=1}^d \EE \sqparen{\mathrm{Var}\big(f_j(X)\big|w(X)\big)}}{d(e^{\frac{2\epsilon}{d}}-1)},\]
with $\cW^\star=\{w \in \cW:\,  \exists\, g \in \cG,\,g \sim w \}$, then $M_\mathsf{G}$ is $\epsilon$-MI PP.
\end{theorem}
The derivation of Theorem~\ref{thm:GaussianMechanismMIPP} (Section \ref{proof:GaussianMechanismMIPP}) and uses the fact that the Gaussian distribution maximizes differential entropy subject to a second moment constraint. 

\begin{remark}[Comparison of Laplace and Gaussian mechanisms for $\epsilon$-MI PP]
In the high-privacy regime (i.e., when $\epsilon$ is small), we have $e^{\epsilon}-1 =\Theta(\epsilon)$. 
The expected noise variance of Laplace mechanism is $\EE[ \|Z_{\mathsf{L}}\|^2]=\Theta(d/\epsilon^2)$, while the Gaussian mechanism has noise variance $\EE[ \|Z_{\mathsf{G}}\|^2]=\Theta(d/\epsilon)$. The Gaussian $\epsilon$-MI PP mechanism thus injects noise with a lower variance than the Laplace mechanism in this case. 
{This is illustrated in \cref{fig:sub-first} for the setting where $f$ is the average of each column of a database in the space $\{0,1\}^{n \times d}$ with fixed $n=100$ and varying $d$. Specifically, the figure shows Laplace and Gaussian noise variance needed to achieve $\epsilon$-MI DP for $d=1,2,5,10,30$.}
\end{remark}

\begin{corollary}[Gaussian mechanism for DP]\label{cor:GaussianMechanismMIDP} 
Under the setup of Theorem~\ref{thm:GaussianMechanismMIPP}, Gaussian mechanism with 
\[\sigma^2 \geq 
 \sup_{\substack{P_X \in \cP(\cX^{n \times k}),\\i=1,\ldots,n}}
\frac{ \sum_{\ell=1}^d \EE \sqparen{\mathrm{Var}\big(f_\ell(X)\big| \big(X(j,\cdot)\big)_{j\neq i}\big)}}{d(e^{\frac{2\epsilon}{d}}-1)},\] is $\epsilon$-MI DP. Furthermore, the the statement remains true if the RHS above is replaced with $\frac{\Delta^2_2(f)}{2d(e^{\frac{2\epsilon}{d}}-1)}$, where $\Delta_2(f):=\max_{x\sim x'} \| f(x) -f(x')\|$. Additionally, if $\mathcal{X}$ is compact and $f: \mathcal{X}^{n \times k}\to \RR $ is continuous, then
$\sigma^2 \geq \Delta^2_2(f)/\big(4(e^{{2\epsilon}}-1)\big)$ is sufficient to achieve $\epsilon$-MI DP.

\end{corollary}

\begin{remark}[Classic Gaussian mechanisms for DP]

By Theorem \ref{thm:equivalence} with the condition $\Theta= \cP(\cX^{n \times k})$, which holds in the DP setting, we have that $\epsilon$-MI DP implies regular $(\epsilon',\sqrt{2\epsilon})$-DP, for any $\epsilon'\geq 0$ . For comparison, the classical Gaussian mechanism achieves $(\epsilon',\sqrt{2\epsilon})$-DP with $\sigma^2\geq2\log(1.25/ \sqrt{2\epsilon}) \Delta^2_2(f)/ {\epsilon'}^2$ for $\epsilon^\prime \leq 1$~\cite{DR14}. It can be shown that our mechanism requires a lower noise level than the classic one whenever $\epsilon' < 2 \big(d(e^{2\epsilon/d}-1) \log(1.25/\sqrt{2\epsilon})\big)^{1/2}$. {\cref{fig:sub-second} shows the region of $(\epsilon',\epsilon)$ values for which our mechanism injects noise of a lower variance for $d=30$. We also note that this region is monotonically decreasing (in the sense of inclusion) with $d$.}

\begin{figure}[!t]
\begin{subfigure}{.5\textwidth}
  \centering
  \includegraphics[width=.8\linewidth]{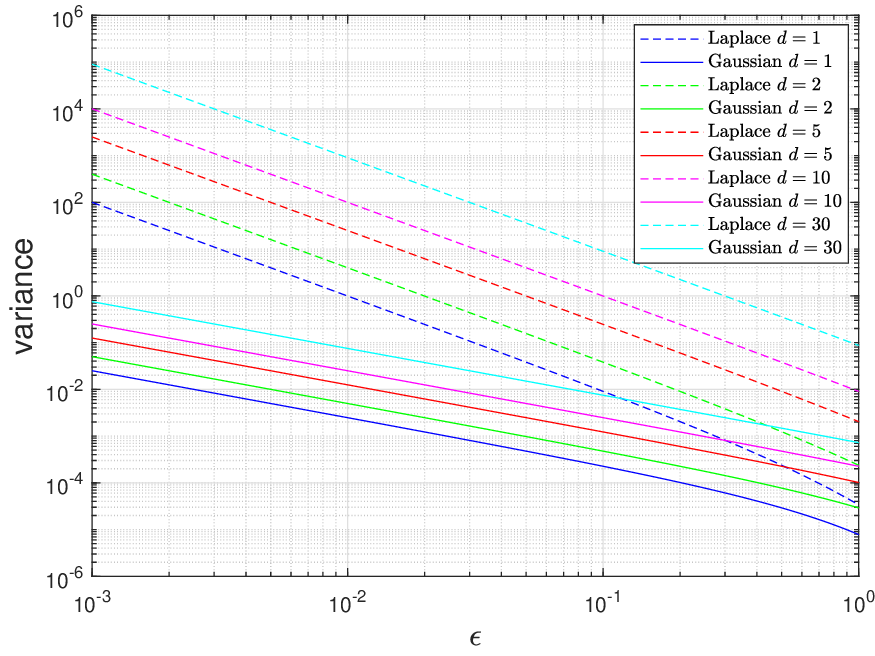}  
  \caption{}
  \label{fig:sub-first}
\end{subfigure}
\hspace{4mm}
\begin{subfigure}{.5\textwidth}
  \centering
  \includegraphics[width=.8\linewidth]{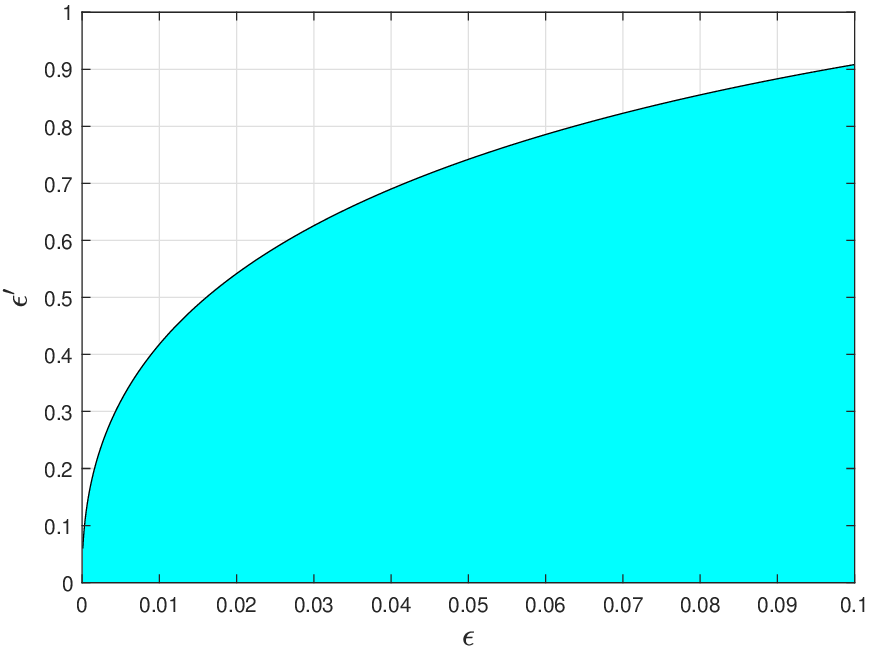}  
  \caption{}
  \label{fig:sub-second}
\end{subfigure}
\vspace{2mm}
\caption{(a) Laplace and Gaussian noise variance injected to achieve $\epsilon$-MI DP for the following setting where $f$ is the average of each column of the database in the space $\{0,1\}^{n \times d}$ with fixed $n=100$ and varying $d=1,2,5,10,30$. (b) The region where the MI DP Gaussian noise injection mechanism adds noise with smaller variance compared to classical mechanism for achieving $(\epsilon',\sqrt{2 \epsilon})$-DP with $d=30$}\label{fig:region Better}
\end{figure}

\end{remark}

The noise levels derived in Theorem~\ref{thm:GaussianMechanismMIPP} scales linearly with the increasing dimension of the query. This may result in large noise values which may compromise utility. 
Projecting the high-dimensional queries to a low-dimensional space may provide a better privacy-utility tradeoff if the dimension of the projection space is chosen appropriately. 

\begin{theorem}[Gaussian mechanism with projections] \label{Thm: Gaussian mechanism with slicing and noise injection}
Fix $\epsilon>0$, a structured PP framework $(\cG,\cW,\cE,\Theta)$. Let $f: \mathcal{X}^{n\times k} \to \RR^d $ be the query function, $\rA\in\RR^{d\times \ell}$  be a projection matrix with $\ell \leq d$, and consider the Gaussian mechanism $M^{\mathsf{proj}}_\mathsf{G}(X):=\rA^\intercal f(X)+Z_\mathsf{G}$, where $Z_\mathsf{G} \sim \cN(0,\sigma^2 \mathrm{I}_\ell)$. Then $M^{\mathsf{proj}}_\mathsf{G}$ is $\epsilon$-MI PP if either of the following conditions hold:
\begin{enumerate}[wide, labelindent=10pt]
    \item $\rA=[\phi_1, \ldots, \phi_\ell]$ is deterministic and
    \[\sigma^2 \geq \sup_{P_X \in \Theta,\, w \in \cW^\star} \frac{ \EE \left[\|\Sigma_{f|w} \|_{\mathrm{op}}\right] \max_{1 \leq j \leq \ell} \|\phi_j\|^2 }{(e^{\frac{2\epsilon}{\ell}}-1)},\]
    where $\|\cdot\|_{\mathrm{op}}$ is the operator norm, $\Sigma_{f|w}$ is the conditional covariance matrix of $f(X)$ given $w(X)$, and $\cW^\star$ is as in Theorem~\ref{thm:LaplaceMechanismMIPP}.
    \item $\rA=[\Phi_1, \dots, \Phi_\ell]$ is random, chosen independently of the database $X$ and the mechanism $M$, with $\EE[\|\Phi_j\|^2] =1$ and $\ \EE[\Phi_j]=\mathbf{0}$  for all $j=1,\ldots,\ell$, and
    \[\sigma^2 \geq  \sup_{P_X \in \Theta,\, w \in \cW^\star}  \frac{ \EE \left[\|\Sigma_{f|w} \|_{\op} + \|\mu_{f|w}\|_2^2 \right]  }{(e^{\frac{2\epsilon}{\ell}}-1)},\]
    where $\mu_{f|w}:=\EE\big[f(X)\big|w(X)\big]$.
\end{enumerate}
\end{theorem}

\begin{remark}[Gaussian projection matrix]
The random matrix whose entries are sampled independently from the Gaussian distribution with 0 mean and $1/d$ variance satisfies the requirements of Theorem~\ref{Thm: Gaussian mechanism with slicing and noise injection}, Part (2). A similar approach was proposed in \cite{ProjectionGaussiankenthapadi2012privacy} for the task of estimating distances between users without being leaking private information. Their $\epsilon$-DP mechanism first projected the query onto a random lower-dimensional space via Johnson-Lindenstrauss transform and then injected Gaussian noise. Theorem~\ref{Thm: Gaussian mechanism with slicing and noise injection} thus enables using $\epsilon$-MI PP in such settings. 
\end{remark}
\begin{remark}[Projection dimension] 
For a $d$-dimensional query, the Gaussian mechanism without projections (Theorem~\ref{thm:GaussianMechanismMIPP}) adds noise proportional to $d$, which may be prohibitive when $d\gg 1$. Theorem \ref{Thm: Gaussian mechanism with slicing and noise injection} shows that by incorporating a projection matrix, one may inject noise that is proportional to $\ell$, where $\ell \ll d$. In practice, the projection dimension $\ell$ should be tuned to optimize the privacy-utility tradeoff, keeping in mind that larger $\ell$ would typically necessitate larger $\sigma^2$ values to guarantee privacy at a prescribed level~\cite{ProjectionGaussiankenthapadi2012privacy}. 
\end{remark}

\subsection{Mechanisms with Explicit Dependence on Private Functions}

The noise levels derived in Theorem \ref{thm:GaussianMechanismMIPP} and \ref{thm:LaplaceMechanismMIPP} depend on the private function class $\cG$ only through $\cW^\star$. However, it may be desirable to capture the dependence on $\cG$ more explicitly. This is particularly relevant when there is no public information (e.g., the AP framework from \cite{ZOC20}; cf. Remark~\ref{rem:pp_reductions} Part 2) or if there is a single public function $w$ corresponding to all private $g\in\cG$. The following theorem provides noise levels with such explicit dependence.

\begin{theorem}[Gaussian $\mspace{-5.3mu}$mechanism$\mspace{-5.3mu}$ with$\mspace{-5.3mu}$ dependence$\mspace{-5.3mu}$ on~$\mspace{-5.3mu}\cG$] \label{Thm:GaussianMechanismMIPP_with entropy law}
Under the setup from Theorem \ref{thm:GaussianMechanismMIPP} and assuming $\inf_{\substack{g\in\cG,w\in\cW:\\ g\sim w}}\sh\big(f(X)\big|g(X),w(X)\big)>-\infty$, the Gaussian mechanism $M_\mathsf{G}$ achieves $\epsilon$-MI PP, if
\[\sigma^2 \geq \sup_{\substack{P_X \in \Theta,\\g\in\cG,w\in\cW:\\g \sim w}} \frac{ A - d \, e^{2\epsilon/d} B}{d ( e^{2\epsilon / d}-1)}\vee 0\,,\]
with $A=\sum_{j=1}^d \EE\big[\mathrm{Var}\big(f_j(X)\big|w(X)\big)\big]$ and $B= \frac{1}{2\pi}\exp\big(\frac{2}{d} \sh\big(f(X) \big|g(X),w(X)\big)-1\big)$.
\end{theorem}
In addition to maximum entropy arguments, the proof of Theorem~\ref{Thm:GaussianMechanismMIPP_with entropy law} uses the entropy power inequality. We may replace the conditional entropy in $B$ by any lower bound that may be easier to compute (cf. e.g., \cite{marsiglietti2018lower}), and $\epsilon$-MI PP will still~hold.

\begin{remark}[Free $\epsilon$-MI PP regime] \label{rem:free_MIPP_regime}
The bound in Theorem~\ref{Thm:GaussianMechanismMIPP_with entropy law} suggests that if $A \leq d \, e^{2\epsilon/d} B$ over the entire optimization domain, $\epsilon$-MI PP holds without noise injection (i.e., $\sigma=0$). It can be shown that $A \geq dB$ for any $P_X\in\cP(\cX^{n \times k})$ and functions $f$, $g$, and $w$. The free privacy regime therefore corresponds to cases where $\epsilon$ is large compared to $d/2$. Since large $\epsilon$ values are rarely of interest in practice, we conclude that a positive noise level is generally needed for $\epsilon$-MI PP. For fixed $\epsilon$ and $d$, the above condition is related to how correlated the query and the private functions are, given the public information. For instance, if $d=1$ and $f(X)$, $g(X)$, and $w(X)$ are jointly Gaussian, we have $A \leq d \, e^{2\epsilon/d} B$ whenever the conditional correlation coefficient between $f(X)$ and $g(X)$ given $\{w(X)=c\}$ satisfies
$ \rho\big(f(X),g(X)\big|w(X)=c \big) \leq \sqrt{ (e^{2\epsilon}-1)e^{-2\epsilon}}.$
Accordingly, weak correlation may lead to free privacy since the query leaks little information about the secret to begin with. Proofs related to the above arguments are presented in Section~\ref{sec:Proofs related to Remark free MI PP} 
\end{remark}

A Gaussian mechanism with explicit dependence on the secret functions was proposed for AP in \cite{ZOC20}, under a rather stringent setting. In their AP formulation there are not public functions (i.e., $\cW=\emptyset$) and taking $d=1$, they assume that $f(X)$ conditioned on $g(X)$ is Gaussian with a constant variance, i.e., $\mathrm{Var}\big(f(X)\big|g(X)=a\big)=\mathrm{Var}\big(f(X)\big|g(X)=b\big)$, for all $a,b \in \mathrm{Im}(g)$. Theorem 1 of \cite{ZOC20} then shows that $(\epsilon,\delta)$-AP is achieved by the Gaussian mechanism with
\[\sigma^2\geq \sup_{\substack{P_X \in \Theta,\\g\in\cG}}\left[\left(\frac{C\Delta_\mathsf{AP}(f)}{\epsilon}\right)^2- \mathrm{Var}\big(f(X)\big|g(X)=a\big)\right]\vee 0,\]
where $C= \sqrt{2\log(1.25/\delta)}$ and 
\[\Delta_\mathsf{AP}(f)\mspace{-3mu}=\mspace{-3mu} \max_{a,b \in \mathrm{Im}(g)} \mspace{-3mu}\big| \EE\big[f(X)\big|g(X)\mspace{-3mu}=\mspace{-3mu}a\big] \mspace{-2mu}-\mspace{-2mu}  \EE\big[f(X)\big|g(X)\mspace{-3mu}=\mspace{-3mu}b\big] \big|.
\]

By means of comparison, the following corollary specializes our Theorem \ref{Thm:GaussianMechanismMIPP_with entropy law} to the setting from \cite{ZOC20}.

\begin{corollary}[Gaussian mechanism for AP]\label{COR:Gaussian_AP_mech} Under the above setting, the
Gaussian mechanism with the variance parameter
\[
 \sigma^2 \geq \sup_{\substack{P_X \in \Theta,\\g\in\cG}} \left[\frac{ \mathrm{Var}(f(X)) - \ e^{2\epsilon} \mathrm{Var}\big(f(X)\big|g(X)=a\big)}{ e^{2\epsilon}-1}\right]\vee 0
\]
satisfies $\epsilon$-MI attribute privacy.
\end{corollary}

Note that both the mechanisms enter the free privacy regime when $f(X)$ is independent of $g(X)$ (Remark~\ref{rem:free_MIPP_regime} above argues that this holds for our mechanism even when there is a weak dependence between $g(X)$ and $f(X)$). Under a multivariate extension of the product Gaussian family from Remark \ref{rem:LaplaceMIDP_Domain}, i.e., 
\[
   \Theta_\mathsf{G}^k(m,s)=\Big\{\cN(\mu,\Sigma)^{\otimes n}: \ \mu=(\mu_1\ldots\mu_k)^\intercal, |\mu_j| \leq m,  \ \Sigma(i,j) \leq s,\ \forall\, i=1,\ldots,n,\, j=1,\ldots,k \Big \} , 
\]
and for $f$ and $g$ linear functions of the database columns (say, average of the column entries), $\Delta_\mathrm{AP}(f)$ is proportional to $\max_{a,b \in \mathrm{Im}(g)} |a-b|$ and thus diverges to infinity. The variance-based bound from Corollary~\ref{COR:Gaussian_AP_mech}, on the other hand, is finite and feasible.

\section{Auditing for Privacy}\label{Sec: Auditing}

{Privacy auditing aims to detect violations in privacy guarantees, reject incorrect algorithms, and provide counterexamples. This concept has been gaining recent attention for the special case of DP auditing. In \cite{jagielski2020auditing}, an heuristic approach based on poisoning attacks was proposed for auditing privacy violations of DP based stochastic gradient descent. 
The idea of formulating DP auditing as a hypothesis test was originally explored in \cite{ding2018detecting} for univariate queries. An extension to the multivariate query setting was proposed in \cite{domingo2022auditing} by relaxing DP to a privacy notion based on kernel R\'enyi divergence, and using an empirical version of the latter as a test statistic. However, all these approaches are tailored for classical DP and are not applicable beyond that setting, e.g., to PP auditing. 

We propose a hypothesis testing pipeline to audit for $\epsilon$-MI DP which readily extends to the PP setting (see \cref{rem:audit_PP} ahead). From relative strength considerations, our approach can, in turn, audit for any stricter privacy notion, such~as $\epsilon$-KL DP, $(\alpha,\epsilon)$-R\'enyi DP, or $\epsilon$-DP itself.  Indeed, since $\epsilon$-MI DP is a relaxation of $\epsilon$-DP, any mechanism that violates the former must also violate the latter. Our audit tests between the null hypothesis $\cH_0$, that $\epsilon$-MI DP holds, and the alternative $\cH_1$ using an estimate of the mutual information from \eqref{EQ:MI_DP} as the test statistic. If the estimate is larger than the threshold, we reject the null and declare the mechanism as violating $\epsilon$-MI DP (and thus also $\epsilon$-DP). When auditing for MI DP, the only source of error in the decision is the statistical estimation error. When auditing DP, on the other hand, an extra slackness may arise from the gap between the DP constraint (say, in terms of $\infty$-Renyi divergence) and the relaxed mutual information-based one. 
}

The main challenge of the proposed approach lies in the inherent hardness of estimating mutual information. For continuous, high-dimensional random variables (which is often the regime of interest in modern privacy applications), the sample complexity of estimating mutual information grows exponentially with dimension \cite{Paninski2003,mcallester2020formal}, making tests based on $\epsilon$-MI DP infeasible. Fortunately, for auditing purposes we may further relax the $\epsilon$-MI DP privacy notion in order to gain sample efficiency of the test. To that end, we propose to employ sliced mutual information (SMI). {Note that this relaxation also comes at a cost in terms of test power due to the gap between mutual information and SMI.}

\subsection{Sliced Mutual Information}

SMI was introduced in \cite{goldfeld2021sliced} as an information measure that preserves many properties of classic (Shannon) mutual information, while being amenable to scalable estimation from high-dimensional samples. SMI is defined as an average of mutual information terms between one-dimensional projections of the considered random variables, namely, for $(X,Y)\sim P_{XY}\in \cP(\RR^{d_x}\times\RR^{d_y})$ it is given by
\begin{equation}
     \smi(X;Y):=\int_{\unitsphx}\int_{\unitsphy}\sI(\theta^\intercal X;\phi^\intercal Y)d\sigma_{d_x}(\theta)d\sigma_{d_y}(\phi),\label{EQ:SMI}
\end{equation}
where $\unitsph:=\{x\in\RR^d:\,\|x\|=1\}$ is the unit sphere in $\RR^d$ and $\sigma_d$ is the uniform distribution on it (cf. \cite{goldfeld2022kSMI} for an extension to $k$-dimensional projections). Proposition~1 of \cite{goldfeld2021sliced} shows that SMI satisfies many properties akin to classic mutual information, such as identification of independence, (sliced) entropy-based decompositions, tensorization, variational forms, and more. By the data processing inequality, SMI is always upper bounded by classic mutual information, i.e., $\smi(X;Y)\leq \sI(X;Y)$, which enables using it to define a relaxed privacy notion.

We also recall the sliced entropy and conditional SMI. For $(X,Y,Z)\sim P_{XYZ}\in\cP(\RR^{d_x}\times\RR^{d_y}\times \RR^{d_z})$ and $(\Theta,\Phi,\Psi)\sim \sigma_{d_x}\otimes\sigma_{d_y}\otimes\sigma_{d_z}$, the sliced entropy of $X$, its conditional version given $Y$, and the condition SMI between $X$ and $Y$ given $Z$ are defined, respectively, by
\begin{align*}
    \ssh(X)&:=\sh(\Theta^\intercal X|\Theta)\\
    \ssh(X|Y)&:=\sh(\Theta^\intercal X|\Theta,\Phi,\Phi^\intercal Y)\\
    \smi(X;Y|Z)&:=\sI(\Theta^\intercal X;\Phi^\intercal Y|\Theta,\Phi,\Psi,\Psi^\intercal Z).
\end{align*}
With these definitions, we have $\smi(X;Y)=\ssh(X)-\ssh(X|Y)$ and $\smi(X;Y,Z)=\smi(X;Y)+\smi(X;Z|Y)$, among others.

\subsection{Sliced Mutual Information Differential Privacy}

A randomized mechanism $M$ is said to satisfy $\epsilon$-SMI DP (w.r.t. the distribution class $\Theta\subseteq \cP(\cX^{n\times k})$ if 
\begin{equation}
    \sup_{\substack{P_X \in \Theta,\\i=1,\ldots,n}} \smi\big(X(i,\cdot);M(X)\big|\big(X(j,\cdot)\big)_{j\neq i}\big) \leq \epsilon;\label{EQ:SMI_DP}
\end{equation}
here we unfold $\big(X(j,\cdot)\big)_{j\neq i}$ into a vector of size $k(n-1)$ and then project it. Evidently, this is similar to the definition of $\epsilon$-MI DI from \cite{CY16} (Definition \ref{def:MI DP}) but with SMI replacing MI and while allowing for distributional assumptions on the database (as in the structured PP framework). We henceforth make the simplifying assumption that the database is independent across records. Namely, denoting the marginal distribution of the $i$th row $X(i,\cdot)$ by $P_i$, we assume that $X\sim P_X=\prod_{i=1}^n P_i$. While the subsequent results are derived under this restriction, we expect the ideas to naturally extend to general data distributions and PP frameworks. The following proposition states the $\epsilon$-SMI DP is a relaxation of $\epsilon$-MI DP.

\begin{proposition}[$\epsilon$-MI DP relaxation]
If a randomized mechanism $M:\cX^{n\times k}\to \cY$ is $\epsilon$-MI DP then it is also $\epsilon$-SMI~DP.
\end{proposition}

The proof is immediate since under the independence assumption (and hence it is omitted), we have
\begin{align*}
\smi\big(X(i,\cdot);M(X)\big|&\big(X(j,\cdot)\big)_{j\neq i}\big)\\
&=\smi\big(X(i,\cdot);M(X),\big(X(j,\cdot)\big)_{j\neq i}\big)\\
&\leq \sI\big(X(i,\cdot);M(X),\big(X(j,\cdot)\big)_{j\neq i}\big)\\
&= \sI\big(X(i,\cdot);M(X)\big|\big(X(j,\cdot)\big)_{j\neq i}\big).
\end{align*}
Thus, we deduce that if $M$ violates $\epsilon$-SMI DP then it cannot be $\epsilon$-MI DP nor $\epsilon$-DP. In fact, it suffices to find a single distribution $P_X\in\Theta$ for which \eqref{EQ:SMI_DP} does not hold to reject the null hypothesis and declare violation of $\epsilon$-DP. These are the main observations for devising the SMI-based audit for DP (Section \ref{SUBSEC:SMI_audit}). Another key component of the test is scalability with which SMI can be estimated, which is formulated next.

\subsection{Sliced Mutual Information Estimation}

We consider estimating of SMI statistic from the left-hand-side (LHS) of \ref{EQ:SMI_DP}, but for a fixed distribution $P_X\in\Theta$ (a further relaxation). Defining the shorthands $X_i=X(i,\cdot)$, $Y=M(X)$, and $Z_i=\big(X(j,\cdot)\big)_{j\neq i}$, our objective is thus $\max_{i=1,\ldots,n}\smi_i$, where $\smi_i:=\smi(X_i;Y|Z_i)$. We first describe a Monte Carlo based estimation procedure for SMI, employing a generic mutual information estimator between scalar variables. Afterwards, we instantiate the generic estimator via the neural estimation framework of \cite{sreekumar2022neural} and provide formal convergence guarantees.

\medskip
\textbf{Monte Carlo estimate of SMI.} Fix $i=1,\ldots,n$, let $\{(X^j_i,Y^j,Z^j_i)\}_{j=1}^m$ 
be $m$ i.i.d. samples of $(X_i,Y,Z_i)$,
and proceed as follows:
\begin{enumerate}[leftmargin=.45cm]
    \item Draw $\{(\Theta_j,\Phi_j,\Psi_j)\}_{j=1}^p$ i.i.d. projection triples from $\unitsphxi \otimes \unitsphd \otimes \unitsphxR$.
    \item For each $j=1,\ldots,m$ and $\ell=1,\ldots,p$, compute $(\Theta_\ell^\intercal X_i^j, \Phi_\ell^\intercal Y^j, \Psi_\ell^\intercal Z_i^j)$. 
    \item For each $\ell=1,\ldots,p$, we estimate the mutual information $\sI\big(\Theta_\ell^\intercal X_i;\Phi_\ell^\tr Y, \Psi_\ell^\tr Z_i\big)$ using the estimate $\hat{\sI}\big((\Theta_\ell^\tr X_i)^m,(\Phi_\ell^\tr Y)^m,(\Psi_\ell^\tr Z_i)^m \big)$ (to be described shortly), where $$(\Theta_\ell^\tr X_i)^m:=(\Theta_\ell^\tr X^1_i,\ldots,\Theta_\ell^\tr X^m_i)$$ and $(\Phi_\ell^\intercal Y)^m$,  $(\Psi_\ell^\tr Z_i)^m $ are defined similarly.
    \item Take a Monte-Carlo average of the above estimates, resulting in the SMI estimator:
\begin{equation}
\widehat{\smi}^{m,p}_i:=\frac{1}{p}\sum_{\ell=1}^p \hat{\sI}\big((\Theta_\ell^\tr X_i)^m,(\Phi_\ell^\tr Y)^m,(\Psi_\ell^\tr Z_i)^m \big) .\label{EQ:SMI_est_generic}
\end{equation}
\item Set the SMI DP statistic estimator as 
\begin{equation}
    \widehat{\smi}^{m,p} :=\max_{i=1,\ldots,n} \widehat{\smi}^{m,p}_i. \label{EQ:SMI_est}
\end{equation}
\end{enumerate}  

\textbf{Neural estimation of mutual information.} We now instantiate the generic mutual information estimator $\hat{\sI}(\cdot,\cdot,\cdot)$ in Step 3 via the neural estimation framework of \cite{sreekumar2022neural}, and obtain explicit estimation error bounds in terms of $m$, $p$, and the size of the neural network. The appeal of this approach is twofold: (i) the SMI neural estimator \cite{goldfeld2022kSMI} lower bounds the population objective in the large sample limit, thus serving as a further relaxation which is in line with the auditing pipeline; and (ii) neural estimators are efficiently computable via standard gradient-based optimizers and have low memory footprint, even for high-dimensional data and massive sample sets.

Neural estimation of mutual information relies on the Donsker-Varadhan (DV) variational form, whereby
\[
\sI(U;V)=\sup_{f:\RR^{d_u}\times\RR^{d_v}\to\RR}\EE[f(U,V)]-\log\big(\EE\big[e^{f(\tilde{U},\tilde{V})}\big]\big),
\]
where $(U,V)\sim P_{UV}\in\cP(\RR^{d_u}\times\RR^{d_v})$, $(\tilde{U},\tilde{V})\sim P_U\otimes P_V$, and $f$ is a measurable function for which the expectations above are finite. Given i.i.d. data $(U_1,V_1),\ldots,(U_m,V_m)$ from $P_{UV}$, the neural estimator parameterizes the DV potential $f$ by an $\ell$-neuron shallow network and approximates expectations by sample means,\footnote{Negative samples, i.e., from $P_U\otimes P_V$, can be obtained from the positive ones via $(U_1,V_{\sigma(1)}),\ldots,(U_m,V_{\sigma(m)})$, where $\sigma\in S_m$ is a permutation such that $\sigma(i)\neq i$, for all $i=1,\ldots,m$.}  resulting in the estimate
\[\hat{\sI}_\ell(U^m,V^m):=\sup_{g \in \cG_\ell} \frac 1m \sum_{i=1}^m g(U_i,V_i)- \log\left(\frac 1n \sum_{i=1}^m e^{g(U_i,V_{\sigma(i)})}\right),\]
where the neural network function class is defined as
\[
\cG_\ell(a) 
:=\left\{ g:\mathbb{R}^{3} \rightarrow \mathbb{R}:\begin{aligned}
    & g(z)=\sum\nolimits_{i=1}^\ell \beta_i \phi\left(\langle w_i, z\rangle+b_i\right)+\langle w_0, z\rangle + b_0, 
    \\& 
    \max_{1 \leq i \leq \ell}\|w_{i}\|_1 \vee |b_i| \leq 1,~  
 \max_{1 \leq i \leq \ell}|\beta_i| \leq \frac{a}{2\ell},~ |b_0|,\|w_0\|_1 \leq a \end{aligned}\right\},
\]

with $\phi(z)=z\vee0$ as the ReLU activation, and we set the shorthand $\cG_\ell=\cG_\ell(\log \log\ell \vee 1)$. Inserting $\hat{\sI}_\ell$ with $U^m=(\Theta_\ell^\tr X_i)^m$ and $V^m=\big((\Phi_\ell^\tr Y)^m,(\Psi_\ell^\tr Z_i)^m\big)$ as the generic mutual information estimator in \eqref{EQ:SMI_est_generic}, the neural estimator of $\smi_i$ is
\[\widehat{\smi}_{i,\s{NE}}^{\ell,m,p}:=\frac{1}{p}\sum_{j=1}^p \hat{\sI}_\ell \big((\Theta_j^\tr X_i)^m,(\Phi_j^\tr Y)^m,(\Psi_j^\tr Z_i)^m \big),\]
and the corresponding SMI DP statistic estimator is
$\widehat{\smi}^{\ell,m,p}_\s{NE} :=\max_{i=1,\ldots,n} \widehat{\smi}_{i,\s{NE}}^{\ell,m,p}$. Note the $\widehat{\smi}_{i,\s{NE}}^{\ell,m,p}$ is readily implemented by parallelizing $m$ $\ell$-neuron ReLU nets with inputs in $\RR^{3}$ and scalar~outputs.

\medskip
\textbf{Formal guarantees.} Drawing upon the results of \cite{sreekumar2022neural} for neural estimation of $f$-divergences, we provide non-asymptotic error bounds for $\widehat{\smi}^{\ell,m,p}_\s{NE}$, subject to certain regularity assumptions on $P_{X,M(X)}$. 
Suppose that $ \cX^{n \times k} \subset \RR^{n \times k}$ and $\cY\subset \RR^{d}$ are compact sets and that $P_{XY}$ (recall that $Y=M(X)$) has a Lebesgue density $f_{XY}$ supported on $\cX^{n\times k}\times \cY$. 
For $b,\eta\geq 0$, let $\cF(b,\eta)\subseteq\cP(\cX^{n\times k}\times \cY)$ be the distribution class that contains all $P_{XY}$ as above that also satisfy the following property: 
$\exists\, r \in\sC^{4}_{b}(\cU) \ \textnormal{for some open set }\cU \supset \cX^{n \times k} \times \cY $, such that
$\log f_{XY} = r|_{\cX^{n\times k} \times \cY},$ and $\max_{i=1,\ldots n}\sI(X_i;M(X),Z_i) \leq \eta$ (namely, the log-density has a smooth extension to an open set $\cU$ containing the support). In particular, this class contains distributions whose densities are smooth and bounded from above and below and admit the aforementioned smooth extension condition. 
This includes uniform distributions, truncated Gaussians, truncated Cauchy distributions, etc.

We next provide convergence rates for the SMI DP statistic $\widehat{\smi}^{\ell,m,p}_\s{NE}$, uniformly over the class $\cF(b,\eta)$, characterizing the dependence of the error on $s$, $m$, and $\ell$. To simplify the bound we assume $\|\cX\|=\|\cY\|=1$, i.e., that the feature and the mechanism output spaces are normalized. The results readily extend to arbitrary compact domains.

\begin{proposition}[Neural estimation error]\label{PROP:SMI_NE}
For any $\eta,b\geq 0$, we have 
\[
  \sup_{P_{XY\in\cF(b,\eta)}}\EE\left[ \left| \max_{i=1,\ldots,n} \smi_i -\widehat{\smi}^{\ell,m,p}_\s{NE}\right|\right]  \leq C\,   n^3 k^2 \big(\ell^{-\frac{1}{2}}+m^{-\frac 12}+ p^{-\frac 12}\big),
\]
where $C$ is a constant that depend only on
$\eta$ and $b$. 
\end{proposition} 

\cref{PROP:SMI_NE} follows by bounding 
\[\EE\left[\left|\max_{i=1,\ldots,n}\smi_i-\widehat{\smi}_{\s{NE}}^{\ell,m,p}\right|\right]\leq \sum_{i=1}^n \EE\left[\left|\smi_i-\widehat{\smi}_{i,\s{NE}}^{\ell,m,p}\right|\right],\] and then applying the SMI neural estimation bound from  \cite{goldfeld2022kSMI} to the RHS above (cf. \cite[Proposition 2]{sreekumar2022neural}). Further details are omitted for brevity.

\subsection{Auditing Differential Privacy via $\epsilon$-SMI DP}\label{SUBSEC:SMI_audit}

We now present the hypothesis testing pipeline for auditing $\epsilon$-SMI DP. Notably, violations on the latter also implies a violation of $\epsilon$-DP. We consider a composite hypothesis test between the null $H_0:\ \max_{i=1,\ldots,n} \smi_i \leq \epsilon$ (i.e., $\epsilon$-SMI DP holds), and the alternative $H_1$. Given samples $\smash{\{(X^j_i,Y^j,Z^j_i)\}_{(i,j)=(1,1)}^{(n,m)}}$ from the database and the mechanism, we use the maximized SMI estimator  $\widehat{\smi}^{\ell,m,p}_{\s{NE}}$ as our test statistic. An immediate consequence of \cref{PROP:SMI_NE} and Markov's inequality is the following Type 1 error bound.

\begin{proposition}[Type-I error] \label{prop:type I error} Fix arbitrary $\epsilon,r>0$ and consider the above setup. We have 
\[ \PP\left(  \widehat{\smi}^{\ell,m,p}_{\s{NE}} > \epsilon +r\, \middle| \,H_0\right) \leq C \frac{n^3 k^2}{r} \big(\ell^{-\frac 12}+m^{-\frac{1}{2}}+ p^{-\frac 12}\big),
\]
where $C$ is the constant from \cref{PROP:SMI_NE}.

\end{proposition}
 
Choosing $r$ and the estimator parameters such that the error probability is not significant, for instance, $r\asymp C \frac{n^3k^2}{\alpha}\big(\ell^{-\frac 12}+m^{-\frac 12}+ p^{-\frac 12}\big)$ with $\alpha\in(0,1)$, and $\ell=m=p\asymp n^6k^4$, we obtain an hypothesis test with level $\alpha$ significance. Indeed, under the null, the rejection probability of this test is below $\alpha$. A power analysis of the proposed test (namely, the Type II error) is also of significant interest since it provides guarantees for identifying privacy violating mechanisms. This, however, requires more advanced machinery such as a limit distribution theory for the test statistic using which a power against local alternatives can be quantified. Since a refined statistical analysis of SMI estimators is beyond the scope of this work, we leave the power analysis of the above test for future work.

\begin{remark}[Auditing variants of DP frameworks]
The above hypothesis test can be used to audit for various privacy framework, including $\epsilon$-DP, $(\alpha,\epsilon)$-R\'enyi DP~\cite{mironov2017renyi} with $\alpha \geq 1$, $\epsilon$-MI DP, etc. This is since all these framework are stronger than (and hence imply) $\epsilon$-SMI DP. Furthermore, when the database distribution has finite support and a bounded density that satisfies the conditions from Theorem~\ref{thm:equivalence}, the implication $(\epsilon',\delta')-\textnormal{DP} \implies \epsilon^*-\textnormal{MI DP}$ holds and the hypothesis test can also audit for $(\epsilon,\delta)$-DP algorithms .

\end{remark}

\begin{remark}[Auditing PP frameworks]\label{rem:audit_PP}
These ideas readily extended to auditing of PP frameworks. In particular, when $\cW=\cE=\emptyset$ in the structured PP framework from Definition~\ref{def:specialized_PP}, the above procedure can be adapted even without requiring that the database rows are i.i.d. (as needed in the case of DP).
\end{remark}

\section{Additional Applications} \label{Sec: Additional applications}

\subsection{Private Mean Estimation}\label{sec: Private mean estimation}

Differentially-private mean estimation is a basic private statistical inference tasks, which was widely studied under the classic DP paradigm \cite{kamath2019privately,kamath2020private,karwa2017finite}. 
We now revisit this problem and quantify the sample complexity of $\epsilon$-MI DP multivariate mean estimation. Potential gains of incorporating domain knowledge into the framework are also discussed. 
Let $X\sim P_X\in\cP(\RR^d)$ be $d$-dimensional random variable with mean $\EE[X]=\mu$. Given $n$ i.i.d. samples of $X$, the goal is obtain an accurate estimate $\hat{\mu}_n$ of its mean that also satisfies $\epsilon$-MI DP. We propose Algorithm \ref{alg:d large MI DP} as the procedure for doing~so.

\begin{algorithm}[t!]
\caption{$\epsilon$-MI DP algorithm for mean estimation 
}\label{alg:d large MI DP}
\begin{algorithmic}
\STATE Input: $n$ i.i.d. data samples $X_1,\ldots,X_n$; $\epsilon$; $\beta$
\STATE $m \gets 200 \log(1/\beta) $
\STATE $k \gets \lfloor n/m \rfloor$
\STATE $\sigma^2 \gets \frac{d m^2}{2n^2 \epsilon}$
\FOR{$p=1:m$}
\STATE $\tilde{\mu}_p \gets \frac{1}{k} \big(X_{(p-1)k+1}+\ldots+X_{p k})+Z_p,\quad Z_p\sim\cN(0,\sigma^2\rI_d)$
\ENDFOR
\STATE $\hat{\mu}_n \gets \argmin_{y \in \RR^d} \sum_{p=1}^m \| y- \tilde{\mu}_p\|$ 
\STATE Output: $\hat{\mu}_n$
\end{algorithmic}
\end{algorithm}

\begin{proposition}[Mean estimation under $\epsilon$-MI DP] \label{Prop:High dimensional mean estimation accuracy guarantees}
Fix $\alpha,\beta,\epsilon >0$. Let $X\sim P_X\in\cP(\RR^d)$ have mean $\EE[X]=\mu$ and a bounded absolute second moment $\EE\big[\|X-\mu\|^2\big] < \infty$. Then Algorithm~\ref{alg:d large MI DP} fed with $n \geq n_0$ data samples, where \[n_0 = O\left( \log(1/ \beta) \left(\frac{d}{\alpha^2}  + \frac{d}{\alpha \sqrt{\epsilon}}  \right) \right),\]
outputs an estimated $\hat{\mu}_n$ that satisfies $\epsilon$-MI DP and achieves $\PP\big(\| \hat{\mu}_n -\mu \| \leq \alpha\big)\geq 1- \beta$.
\end{proposition}

Proposition~\ref{Prop:High dimensional mean estimation accuracy guarantees} is proven in Section~\ref{Proof of Prop:High dimensional mean estimation accuracy guarantees}. The argument uses Theorem~\ref{thm:GaussianMechanismMIPP} to derive noise levels under which Algorithm \ref{alg:d large MI DP} attains $\epsilon$-MI DP, and then applies median of means techniques (specifically geometric median)  to improve the accuracy of the estimate. 

\begin{remark}[Sensitivity-based mechanisms]
Private mean estimation under other variants of DP, such as $\epsilon$-DP, $(\epsilon,\delta)$-DP, and $\rho$-concentrated DP~\cite{bun2016concentrated}, typically employs standard, sensitivity-based noise injection mechanisms. However, these require knowledge of an upper bound on the mean, i.e., $R$ such that $\|\mu \| \leq R$. For example, \cite{kamath2020private} propose a mechanism that is initiated using a rough proxy of $R$ which is iteratively refined using the data samples. 
The $\epsilon$-MI DP mechanism proposed herein, on the other hand, does not require boundedness so long that the distribution of interest has finite variance. {This is due to the fact that $\epsilon$-MI DP incorporates domain knowledge regarding the class of distributions, as opposed to the aforementioned DP variants that guarantee privacy in the worst-case (and are hence stricter). }

\end{remark}

\begin{remark}[Computational complexity of Algorithm~\ref{alg:d large MI DP}]
For $d=1$, Algorithm~\ref{alg:d large MI DP} boils down to computing the median of the means obtained from $m$ rounds. In the multivariate case, our algorithm requires evaluating the geometric median, which is a computationally hard problem. Nevertheless, there are near-linear time algorithms for approximate geometric median computation, i.e., with complexity $O\big(dm (\log(1/\alpha))^3 \big)$ where $\alpha >0$ is the approximation parameter\cite{cohen2016geometric}.\footnote{Another approach for computing a multivariate median is the smallest-ball median method, but it is 
computationally intractable~\cite{hsu2016lossMedianSmallestBall}.} 
To reduce the computational burden, one may consider coordinate-wise median, whose complexity is $O\big(dm\log(m)\big)$; the linear dependence on $d$ can further be alleviated by parallelization. However, the sample complexity needed to to achieve the same level of accuracy $\alpha$ using the coordinate-wise median worsens to $n_0 = O\big( \log(d/ \beta)\big(d\alpha^{-2}+d\alpha^{-1}\epsilon^{-1/2}\big) \big)$.
\end{remark}

{ 
\begin{remark}[Sample complexity of $\epsilon$-DP algorithms]
The sample complexity of $\epsilon$-differentially private mean estimation was evaluated in \cite{kamath2020private}, under the assumption that $\|\mu\| \leq R$, for some known $R$, and that the $k\geq 2$ moment of $X\sim P_X$ is bounded. Their result reads as
 \[n_0= O\left(\log(d/\beta) \left( \frac{d}{\alpha^2}+ \frac{d}{\epsilon \alpha^{k/k-1}}+ \frac{d\log(R)}{\epsilon} \right) \right).\]
The achieving algorithm first truncates the data into a bounded region so as to limit the amount of injected noise needed for privacy, and then performs mean estimation in a differentially private manner. To find the truncation range they use an iterative procedure which depends on the known $R$. This approach, however, becomes computationally infeasible when dimension is large. While the $\epsilon$-DP sample complexity bound above grows like $1/\epsilon$, the one corresponding to $\epsilon$-MI DP is on the order of $1/\sqrt{\epsilon}$ (see \cref{Prop:High dimensional mean estimation accuracy guarantees}). Nevertheless, it is important to note that the privacy guarantees provided by these approaches are different, with $\epsilon$-DP being strictly stronger.

\end{remark}

}

\subsection{Algorithmic Stability} 

Conditional mutual information was used in \cite{steinke2020reasoning} to study algorithmic stability and, in turn, generalization of machine learning models. We next recall the setup from \cite{steinke2020reasoning}, outline their main stability results, and demonstrate that an algorithm satisfying $\epsilon$-MI DP is stable in that sense. Consider a (possibly randomized) learning algorithm $A: \cX^{n \times k} \to \cY$ that takes as input $n$ samples, each with $k$ features, and outputs an hypothesis from the class $\cY$ (the precise task or loss function are inconsequential here). Let $Q\in\cP(\cX^k)$ be the data distribution and draw a $2n$-sized dataset $X \sim Q^{\otimes {2n}}$. Let $B=(B_1,\ldots,B_n)\sim \mathrm{Ber}(0.5)^{\otimes n}$ be a $n$-lengthed string of i.i.d. random bits independent of everything else. Using $B$, we define an $n$-sized database $X_B$, which will be fed into the learning algorithm as follows. Set $X_{B}(i,\cdot)=X\big(iB_i+(n+i)(1-B_i)\big)$, for $i=1,\ldots,n$, i.e., the $i$th entry of $X_B$ is $X(i,\cdot)$ of $B_i=1$ and $X(n+i,\cdot)$ otherwise. Note that by symmetry, the distribution of this database is $X_B\sim Q^{\otimes n}$.

The following conditional mutual information measure, abbreviated CMI, quantifies how informative the output of an algorithm is about the selected samples (which are determined by the string $B$), given the entire $2n$-sized original database. 

\begin{definition}[CMI \cite{steinke2020reasoning}]
Under the setup above, the CMI of the algorithm $A$ w.r.t. the data distribution $Q$ is given by
\[\mathsf{CMI}_{Q}(A):=  \sI\big(A(X_B);B |X\big),\]
while its distribution-free CMI is
\[\mathsf{CMI}(A):= \sup_{x \in \cX^{2n \times k}} \sI\big(A(x_B);B \big). \]
\end{definition}

Several generalization bounds based on the above definition were proven in \cite{steinke2020reasoning}, which roughly behave as $\big(\mathsf{CMI}_Q(A)/n\big)^{1/2}$ (or the distribution-free analogue); cf., e.g., Theorems 1.2 and 1.3 therein.  In addition,~\cite{steinke2020reasoning} showed that DP algorithms have bounded distribution-free CMI. Specifically, an algorithm $A$ that satisfies $\sqrt{2\epsilon}$-DP has $\mathsf{CMI}(A)~\leq~\epsilon n$. The following result shows that $\epsilon$-MI DP also entails a similar conclusion.

\begin{proposition}[CMI bound via $\epsilon$-MI DP] \label{cor: Algorithm stability via MI DP}
If the algorithm $A: \cX^{n \times k} \to \cY$ satisfies $\epsilon$-MI DP, then we have $\mathsf{CMI}(A) \leq \epsilon n $.
\end{proposition}
This result follows because $\epsilon$-MI DP algorithms also satisfy $\epsilon$-mutual information stability\footnote{Given a data set $X=(X_1,\ldots,X_n)$ that comprises $n$ i.i.d. samples from $Q\in\cP(\cX^{k})$, an algorithm $A$ is said to be $\epsilon$-mutual information stable if 
$\sup_{Q\in\cP(\cX^{k})} \frac{1}{n} \sum_{i=1}^n \sI \big(A(X);X(i,\cdot) \big| \big(X(j,\cdot)\big)_{j \neq i} \big) \leq~\epsilon.
$}~\cite{raginsky2016information}, which by \cite[Theorem~4.7]{steinke2020reasoning} further implies the CMI bound above. Consequently, learning algorithms that satisfy $\epsilon$-MI DP follow the generalization bounds from \cite{steinke2020reasoning}.

\begin{remark}[Domain knowledge]
Algorithmic stability on restricted classes of distributions may be of interest when information about the data generating process is available. The above ideas and results readily extended to this case by replacing the set of all possible distributions by a subset $\Theta \subset \cP(\cX^{n \times k})$.
\end{remark}

\subsection{Utility of PP and MI PP Mechanisms}\label{app:utility}

Generally, it may be hard to assess the utility of a given PP mechanism. However, the MI PP framework can be used as a proxy to understand the privacy-utility tradeoff in the mechanism design phase. We showcase another scenario here where we can borrow the results connected to information theory in making progress in privacy research via MI PP. In this subsection, we focus on the setup where $\Theta=\{ P_X\}, \cG=\{g\} $ and $\cW= \emptyset$, namely, when the database is drawn from a given distribution and we want to hide one specific function thereof. Let $f(X)$ be the query and $M(X)$ the output of the private mechanism. We assume that the ranges of $f$ and $g$ are finite and cast $\sI(f(X);M(X))$ as the utility metric, i.e., how much information the mechanism's output has of the query. 

Consider the maximal utility of an $\epsilon$-MI PP mechanism $M(X)$ in two situations: when $g(X)$ is or is not available at the mechanism design step. Namely, we define
\begin{align*}
    \cU^\epsilon_1(X) &:= \sup_{\substack{P_{M(X)|f(X),g(X)},\\ \sI(M(X);g(X)) \leq \epsilon}} \sI\big(f(X);M(X)\big)\\
   \cU^\epsilon_2(X)&:= \sup_{\substack{P_{M(X)|f(X)},\\ g(X) \rightarrow f(X) \rightarrow M(X), \\ \sI(M(X);g(X)) \leq \epsilon}} \sI\big(f(X);M(X)\big) 
\end{align*}

The following proposition, which relies on \cite[Lemma~8]{zamani2022bounds}, gives an upper bound on the maximal utility. 

\begin{proposition}[Lemma~8 of \cite{zamani2022bounds}] \label{prop:max_utility}
For any $0 \leq \epsilon < \sI(g(X);f(X))$, 
we have 
\[  \cU^\epsilon_2(X) \leq \cU^\epsilon_1(X) \leq \mathsf{H}\big(f(X)|g(X)\big) + \epsilon.\]
\end{proposition} 

From Proposition~\ref{prop:max_utility}, we observe that the maximum utility which could be expected from any $\epsilon$-MI PP mechanism (regardless of the accessibility to $g(X)$) is at most $\mathsf{H}\big(f(X)|g(X)\big) + \epsilon$. Since $\epsilon$-PP implies $\epsilon$-MI PP, all PP mechanisms are included in the optimization and we obtain a maximal utility bound for them as well. In particular, the bound shows that the increase in utility is at most linear in $\epsilon$.

\medskip
Through the relation to MI PP, we can also show the existence of MI PP mechanisms that achieve a certain utility lower bound. 
\begin{proposition}[Theorem~2 of \cite{zamani2022bounds}]\label{prop:min_utility}
For any $0 \leq \epsilon < \sI(g(X);f(X))$, 
we have 
\[  \cU^\epsilon_1(X) \geq   L_1^\epsilon \vee  L_2^\epsilon,\]
where 
$L_1^\epsilon  = \mathsf{H}\big(f(X)|g(X)\big) -\mathsf{H}\big(g(X)|f(X)\big) + \epsilon$ and 
$L_2^\epsilon= \mathsf{H}\big(f(X)|g(X)\big) -\alpha \mathsf{H}\big(g(X)|f(X)\big) + \epsilon -(1-\alpha) \big( \log( \sI(g(X);f(X))+1) +4\big)$ with $\alpha= \epsilon/ \mathsf{H}(g(X))$.
\end{proposition}

\cref{prop:min_utility} implies that there exists an $\epsilon$-MI PP mechanism with utility of at least $ L_1^\epsilon \vee  L_2^\epsilon$. Notice that when $\mathsf{H}\big(g(X)|f(X)\big)=0$ (i.e., $g(X)$ is deterministic given $f(X)$), this lower bound is tight. Under this setting, combining Propositions \ref{prop:max_utility} and \ref{prop:min_utility}, one may deduce that there exists an $\epsilon$-MI PP mechanism with utility $ \cU^\epsilon_1(X) = \mathsf{H}\big(f(X)|g(X)\big) + \epsilon$. We stress, however, that this is merely an existence claim that does not reveal how to design a mechanism that achieves this maximum utility.

\section{Proofs} \label{Sec:proof}

\subsection{Proof of Theorem \ref{thm:equivalence}}\label{proof:equivalence}
For the first implication, note that $\epsilon$-PP implies that
\[\sup_{\cA} \ \log \paren{\frac{\PP \big(M(X) \in \cA \big| \cR\big) }{\PP \big(M(X) \in \cA \big| \cT\big)}}  \leq \epsilon,\quad \forall \ (\cR,\cT) \in \cQ.\]
The left-hand side above is the infinite order R\'enyi divergence. By monotonicity of R\'enyi divergences w.r.t. their order \cite{EH14}, we have 
$\dkl \paren{P_{M(X)|\cR} \| P_{M(X)|\cT}} \leq \epsilon.$ Then,
\begin{align*}
\sI\big(g(X);M(X)|w(X)\big)&= \EE \sqparen{\dkl\big(P_{M(X)|g(X),w(X)} \big\| P_{M(X)|w(X)}\big)} \\
&\leq \EE \sqparen{\dkl \big(P_{M(X)|g(X),w(X)} \big\| P_{M(X)|g(X)',w(X)}\big)}\numberthis\label{EQ:KL_MI}
\end{align*}
where the inequality uses convexity of KL divergence, with $g(X)'$ as an i.i.d. copy of $g(X)$. Recalling that under the structured PP framework secret pairs are $\big(A_{g,w}(a,c),A_{g,w}(b,c)\big)$, with $A_{g,w}(a,c)=\big\{g(X)=a,\,w(X)=c\big\}$, $\epsilon$-MI PP follows by the KL divergence bound. 

\medskip
Assuming $\Theta= \cP(\cX^{n \times k})$, the second implication follows by the minimax redundancy capacity theorem~\cite{CoverThomas06}, which gives
\[\sup_{\Theta} \sI \big(g(X); M(X)|w(X)=c \big)\leq\epsilon\ \implies\ \inf_Q \max_{a} \dkl\big(P_{M(X)|g(X)=a, w(X)=c}\big\| Q\big)\leq\epsilon.\]
Let $Q^{\star}$ achieve the infimum on the RHS. Since $M$ is $\epsilon$-MI PP by assumption, we have $\dkl \big(P_{M(X)|g(X)=a, w(X)=c} \big\| Q^{\star}\big) \leq \epsilon$, for all $a\in\mathrm{Im}(g)$. Applying Pinsker's inequality 
together with the triangle inequality, we obtain
\[\big\| P_{M(X)|g(X)=a,w(X)=c}-  P_{M(X)|g(X)=b,w(X)=c}\big\|_{\mathsf{TV}} \leq \sqrt{2 \epsilon},\]
which implies that $M$ is $(0,\sqrt{2\epsilon})$-PP and hence $(\epsilon',\sqrt{2\epsilon})$-PP (recall Definition~\ref{def:pp_def} for $\epsilon' >0$.

\medskip
For the third implication, in both part we rely on an adaptation of Property 3 from \cite{CY16} from DP (as considered therein) to PP. Specifically, that $(\epsilon,\delta)\textnormal{-PP}$ implies $(0,\delta^{'})\textnormal{-PP}$, with $\delta^{'}=1-{2(1-\delta)}/{(e^\epsilon +1)}$. Having that, for Part (1), we follow the argument from the proof of \cite[Lemma~3]{CY16} to show that if $\big|\supp\big(M(X)\big)\big|<\infty$ or $\max_{g\in\cG}| \mathrm{Im}(g)|<\infty$, then $(0,\delta)$-PP implies $\epsilon^{\star}$-MI PP with $\epsilon^{\star}$ as stated in Theorem~\ref{thm:equivalence}.

For Part (2), we recall that $(0,\delta^{'})\textnormal{-PP}$ is equivalent to the TV bound \[\big\| P_{M(X)|g(X)=a,w(X)=c}-  P_{M(X)|g(X)=b,w(X)=c}\big\|_{\mathsf{TV}} \leq \delta',\quad \forall (a,b)\in\mathrm{Im}(g),\ c\in\mathrm{Im}(w),\]
and then control the mutual information term of interest by the above TV norm. To obtain the first component of the $\epsilon^\star$ expression, we use the reverse Pinsker inequality from \cite[Theorem~1]{sason2015reversePinsker} to translate the above TV bound into the following bound the KL divergence:
\[
\dkl \big(P_{M(X)|g(X)=a,w(X)=c} \big\| P_{M(X)|g(X)'=b,w(X)=c}\big)\leq 
    \frac{\delta'}{2}\sup_{\substack{P_X \in \Theta,\\(g,w)\in\cG\times\cW:\, g \sim w,\\ a,b \in \mathrm{Im}(g),\, c \in \mathrm{Im}(w)}} \left( \frac{\log \big(\alpha_{a,b,c}^{-1}\big)}{1-\alpha_{a,b,c}} - \beta_{a,b,c} \right)=:\epsilon^\star_1.
\]
Having that, we bound the mutual information as in \eqref{EQ:KL_MI} to obtain $\epsilon^\star_1$-MI PP. For the second component of the $\epsilon^\star$ expression, we use the following Lemma which follows directly from Corollary~12 of \cite{xu2020continuity}.

\begin{lemma} \label{lem:Bounded Densities}
If joint density $f_{M(X),g(X),w(X)}$ exists, then we have 
\begin{align}
   & \sI\big( g(X);M(X)|w(X)=c \big) 
   \leq \gamma(c\,;g,w,P_X) \EE \left[ \|P_{M(X)|w(X)=c} -P_{M(X)|w(X)=c,g(X)} \|_{\tv} \right], \label{eq: bounded densities MI bound}
\end{align}
where 
\begin{align*}
    \gamma(c\,;g,w,P_X)
    &=\sup_{a \in \mathrm{Im}(g)} \log\left(\frac{\sup_{z \in \supp(M(X))}f_{M(X),g(X),w(X)}(z,a,c)}{\inf_{z' \in \supp(M(X))}f_{M(X),g(X),w(X)}(z',a,c)} \right). 
\end{align*}

\end{lemma}
\cref{lem:Bounded Densities} together with the fact that $(\epsilon,\delta)\textnormal{-PP} \implies (0,\delta^{'})\textnormal{-PP}$ implies MI PP with parameter \[\epsilon^\star_2:=\delta'\sup_{\substack{P_X \in \Theta,\\(g,w)\in\cG\times\cW:\, g \sim w,\\c\in\mathrm{Im}(w)}}\gamma(c\,;g,w,P_X).\]
Taking $\epsilon^\star=\epsilon^\star_1\wedge \epsilon^\star_2$ yields the result.

\subsection{Proof of Theorem \ref{thm:MIPP_properties}}\label{proof:MIPP_properties}

For (1), let $I$ be a $k$-ary categorical random variable with the probabilities $p_1,\ldots,p_k$. Then
\begin{align*}
  \mathsf{I}\big(g(X);M(X)\big|w(X)\big)&\leq  \mathsf{I}\big(g(X);M(X),I\big|w(X)\big)\\
  &\stackrel{(a)}= \mathsf{I}\big(g(X);M(X)\big|w(X),I\big) \\
  &\stackrel{(b)}=\sum_{i=1}^k p_i \ \mathsf{I}\big(g(X);M_i(X)\big|w(X)\big) 
\end{align*}
where (a) and (b) follow from the independence of $I$ and  $\big(X,M_1(X),\ldots, M_k(X)\big)$. The claim now follows since $M_1,\ldots, M_k$ satisfy $\epsilon$-MI PP.

\medskip
Part (2) directly follows from the chain rule of mutual information \[\mathsf{I}\big(g(X);M^k\big|w(X)\big) = \sum_{i=1}^k \mathsf{I}\big(g(X);M_i(X)\big|w(X),M_1,..,M_{i-1}\big),\]
while Part (3) is a consequence of the data processing inequality \[\mathsf{I}\big(g(X);A(M(X))\big|w(X)\big) \leq \mathsf{I}\big(g(X);M(X)\big|w(X)\big).\]

\subsection{Proof of Proposition~\ref{thm:general_nonAdaptivityMIPP}} \label{Proof:thm:general_nonAdaptivityMIPP}
First use induction to prove that 
\begin{equation}
 \sI\big(g(X);M^k(X)|w(X)\big) \leq \sum_{i=1}^m \epsilon_i + \eta',\label{eq:induction_goal}
\end{equation}
where $\eta'= \sum_{i=2}^k \sI\big(M_i(X);M_1(X),\ldots,M_{i-1}(X)|w(X),g(X)\big)$. Given this inequality, supremizing over $P_X\in\Theta$ and $\cG\ni g\sim w\in\cW$ yields the result of \cref{thm:general_nonAdaptivityMIPP}. 

For $k=2$, consider
\begin{align}
    \sI\big(M_1(X),M_2(X);g(X)|w(X)\big)
    \nonumber & = \sI\big(M_1(X);g(X)|w(X)\big) + \sI\big(M_2(X);g(X)|w(X),M_1(X)\big) \\
    & \leq \epsilon_1 + \sI\big(M_2(X);g(X)|w(X),M_1(X)\big), \label{k=2 proof}
\end{align}
where the last step uses the fact that $M_1$ is $\epsilon_1$-MI PP. For the second above, we have
\begin{align*}
    \sI\big(M_2(X);g(X)&\big|w(X),M_1(X)\big)\\
    &=\sh\big(M_2(X)|w(X),M_1(X) \big) -\sh\big(M_2(X)|w(X),M_1(X), g(X) \big) \\
    & \leq \sh\big(M_2(X)\big|w(X) \big) -\sh\big(M_2(X)\big|w(X),M_1(X), g(X) \big) 
    \pm \sh\big(M_2(X)\big|w(X),g(X) \big) \\ 
    & = \sI\big(M_2(X);g(X)\big|w(X)\big) + \sI\big(M_1(X);M_2(X)\big|w(X),g(X)\big) \\
    & \leq \epsilon_2 + \sI\big(M_1(X);M_2(X)\big|w(X),g(X)\big) ,
\end{align*}
where the last step is since $M_2$ is $\epsilon_2$-MI PP. Collecting the terms, proves the claim for $k=2$. 

\medskip
Assume that the result holds for $k=m$, i.e., 
\begin{align}
  \sI\big(g(X);M^m(X)|w(X)\big) \leq \sum_{i=1}^m \epsilon_i + \eta'' \label{k=m case}
\end{align} 
with $\eta'' = \sum_{i=2}^m \sI\big(M_i(X);M_1(X),\ldots,M_{i-1}(X)|w(X),g(X)\big)$, and consider the following for $k=m+1$:
\begin{align*}
     \sI\big(g(X);M^{m+1}(X)|w(X)\big) &= \sI\big(g(X);M^{m}(X)|w(X)\big) + \sI\big(M_{m+1}(X); g(X)| M^m(X), w(X)\big) \\
     & \stackrel{(a)}\leq \sum_{i=1}^m \epsilon_i + \eta'' + \sI\big(M_{m+1}(X); g(X)| M^m(X), w(X)\big) \\
     & \stackrel{(b)}\leq \sum_{i=1}^m \epsilon_i + \eta'' + \epsilon_{m+1} +  \sI\big(M_{m+1}(X); g(X)| M^m(X), w(X)\big),
\end{align*}
where (a) uses the induction assumption, while (b) follows since $M_{m+1}$ is $\epsilon_{m+1}$-MI PP. This establishes \eqref{eq:induction_goal} and the proof is concluded by supermizing the RHS over $P_X\in\Theta$ and $\cG\ni g\sim w\in\cW$.

{
\subsection{Proof of \cref{prop:Universal composability of MI PP}} \label{proof: Universal composability of MI PP}
\underline{Part (i):} We prove this part by induction. Fix $P_X \in \Theta$ and  $\cG\ni g\sim w\in\cW$. For $k=2$, we use the chain rule together with the fact that $M_1$ is $\epsilon_1$-MI PP to first obtain
\[\sI\big(g(X);M_1(X),M_2(X)\big|w(X) \big)  \leq\epsilon_1 + \sI\big(g(X);M_2(X)\big|w(X),M_1(X) \big).\]
Then, notice that under the assumption that $\Theta\subseteq\Theta_\mathsf{UC}$, for any $P_X\in\Theta$, we have
\begin{align*}
  \sI\big(g(X);M_2(X)\big|w(X),M_1(X) \big)  & \leq \sh\big(M_2(X)\big|w(X)\big) -  \sh\big(M_2(X)\big|w(X), g(X),M_1(X)\big) \\
  &=   \sI\big(g(X);M_2(X)\big|w(X)\big)\\
  &\leq \epsilon_2,
\end{align*}
where the penultimate equality follows since $M_2(X)\leftrightarrow\big(g(X),w(X)\big)\leftrightarrow M_1(X)$ forms a Markov chain whenever $P_X$ is a UC distribution, while the last step is due to $M_2$ satisfying $\epsilon_2$-MI PP. Combining the above, the statement for $k=2$ follows.

\medskip
Assume that for $k=m$, we have 
\[\sI\big(g(X);M^m(X)\big|w(X)\big) \leq \sum_{i=1}^m \epsilon_i.\]
For $k=m+1$, consider 
\[
   \sI\big(g(X);M^{m+1}(X)\big|w(X)\big)  =  \sI\big(g(X);M^{m}(X),M_{m+1}\big|w(X)\big) \leq   \sum_{i=1}^m \epsilon_i + \epsilon_{m+1},
\]
where the inequality follows from the $k=2$ case applied to the mechanisms $(M^m,M_{m+1})$. By induction, we deduce that 
\[\sI\big(g(X);M^k(X)|w(X)\big) \leq \sum_{i=1}^k \epsilon_i,\] 
and the claim follows by supremizing the LHS over $P_X \in \Theta$ and  $\cG\ni g\sim w\in\cW$.

\medskip
\underline{Part (ii):} We again use induction. Fix $P_X \in \Theta$ and  $\cG\ni g\sim w\in\cW$, and for $k=2$ first note that
\begin{align*}
    &\sI\big(g(X);M_1(X),M_2(X)\big|w(X)=c\big) \\
    &\leq \int  \mspace{-3mu}\dkl  \big(P_{M_1(X),M_2(X)|g(X)=a,w(X)=c} \big\| P_{M_1(X),M_2(X)|{g}(X)=\tilde{a},w(X)=c} \big)  \dd P_{g(X)|w(X)}(a |c)  \dd P_{g(X)|w(X)}(\tilde{a} |c)
\end{align*}
which follows by convexity of the KL divergence. Observe that 
\[P_{M_1(X),M_2(X)|g(X),w(X)}(\cdot|a,c)=\int_\cX P_{M_1(X),M_2(X)|X}(\cdot|x)\dd P_{X|g(X),w(X)}(x|a,c)
\]
(which uses the conditional independence of $\big(M_1(X),M_2(X)\big)$ from $\big(g(X),w(X)\big)$ given $X$ itself) and leverage convexity once more to bound
\begin{align*}
    &\dkl  \big(P_{M_1(X),M_2(X)|g(X)=a,w(X)=c} \big\| P_{M_1(X),M_2(X)|{g}(X)=\tilde{a},w(X)=c} \big)\\
    &\leq \int  \mspace{-3mu}\dkl  \big(P_{M_1(X),M_2(X)|X=x} \big\| P_{M_1(X),M_2(X)|X=\tilde x} \big)  \dd P_{X|g(X),w(X)}(x|a,c)  \dd P_{X|g(X),w(X)}(\tilde{x}|\tilde{a},c)\\
    &\leq \int  \Big[\dkl\big(P_{M_1(x)}\big\|P_{M_1(\tilde x)}\big)+\dkl\big(P_{M_2(x)}\big\|P_{M_2(\tilde x)}\big)\Big]\dd P_{X|g(X),w(X)}(x|a,c)  \dd P_{X|g(X),w(X)}(\tilde{x}|\tilde{a},c)
\end{align*}
where the last step is due to the conditional independence of $M_1(X)$ and $M_2(X)$ given $X$ and tensorization of KL divergence. Recall that by definition of UC distributions, for any $(x,\tilde x)\in\supp(P_{X|g(X)=a,w(X)=c})\times \supp(P_{X|g(X)=\tilde a,w(X)=c})$ we have $\lambda \delta_x+(1-\lambda) \delta_{\tilde x}\in\Theta_{\mathsf{UC}}$ for $\lambda \in(0,1)$. Since Part (ii) assumes $\Theta_{\mathsf{UC}}\subseteq\Theta$, the fact that $M_i$, for $i=1,2$, is $\epsilon_i$-PP in the framework with distribution class $\Theta$, we obtain $\dkl\big(P_{M_i(x)}\big\|P_{M_i(\tilde x)}\big)\leq \epsilon_i$. Inserting this into the bounds above yields
\[\sI\big(g(X);M_1(X),M_2(X)\big|w(X)\big) \leq  \epsilon_1 + \epsilon_2,\]
which is the desired claim for $k=2$. 

\medskip
The induction assumption for $k=m$ reads as
\[\sI\big(g(X);M^n(X)\big|w(X)\big) \leq \sum_{i=1}^n \epsilon_i, \]
and for $k=m+1$, we have
\[
    \sI\big(g(X);M^{m+1}(X)\big|w(X)\big) =  \sI\big(g(X);M^{m}(X), M_{m+1}(X)\big|w(X)\big) \leq     \sum_{i=1}^{m+1} \epsilon_i, 
\]
where the inequality follows by the $k=2$ case as before. We conclude again by induction and taking the appropriate supremum.

}

\subsection{Proof of Theorem~\ref{thm:LaplaceMechanismMIPP}}\label{Proof: Laplace mechanism MIPP}
Let $Z_{\mathsf{L}}=(Z_1,\ldots, Z_d)$, with $Z_j\sim \mathrm{Lap}(0,b)$ i.i.d. We have
\begin{align*}
   \sh\big(f(X)+Z_\mathsf{L}\big|w(X)\big)&\stackrel{(a)} \leq \sum_{j=1}^d  \int \sh\big(f_j(X)+Z_j-m_j(c)\big|w(X)=c\big) d P_{w(X)}(c) \\ 
   & \stackrel{(b)}\leq \sum_{j=1}^d \int \log\Big(2e\EE\big[|f_j(X)-m_j(c)+Z_j|\big|w(X)=c\big] \Big)dP_{w(X)}(c)  \\
    & \stackrel{(c)}\leq  d \log \left(2e  \left(\sum_{j=1}^d \frac{1}{d}\, \EE \left[\sqrt{\mathrm{Var}\big(f_j(X) \big|w(X)\big)}\right] + b \right)\right),\numberthis\label{eq:laplace_proof}
\end{align*}
where (a) uses the chain rule, the fact that conditioning cannot increase differential entropy, and its translation invariance with $m_j(c):=~\EE[f_j(X)|w(X)=c]$; (b) is because the Laplace distribution maximizes differential entropy subject to an expected absolute deviation constraint; 
while (c) follows from Jensen's inequality, along with $\EE [|Z_j|]=b$ and $\EE[|X|]^2 \leq \EE[X^2]$.

Combining \eqref{eq:laplace_proof} with $\sh\big(f(X)+Z_{\mathsf{L}}\big|g(X),w(X)\big)\geq \sh(Z_{\mathsf{L}}) = d \log(2be)$ yields: 
\begin{align*}
    \sI \big(g(X);M_\mathsf{L}(X)\big|w(X)\big) &= \sh\big(f(X)+Z_\mathsf{L}\big|w(X)\big)- \sh\big(f(X)+Z_{\mathsf{L}}\big|g(X),w(X)\big) \\
    & \leq d \log \left(2e  \left(\sum_{j=1}^d \frac{1}{d}\, \EE \left[\sqrt{\mathrm{Var}\big(f_j(X) \big|w(X)\big)}\right] + b \right)\right)- d\log (2be). \numberthis \label{Laplace last bound}
\end{align*}
To conclude, further upper bound \eqref{Laplace last bound} by $\epsilon$ and solve for~$b$.

\subsection{Proof of Corollary~\ref{cor:LaplaceMIDP}}\label{Sec: Proof of Cor LaplaceMIDP}
Using the fact that $\mathrm{Var}(Z)=\frac{1}{2} \EE \big[(Z-Z')^2 \big]$ for $Z'$ an i.i.d. copy of $Z$, we have
\begin{align*}
\sum_{j=1}^{d} \mathrm{Var}\left(f_j(X) \middle| \big(X(k,\cdot)\big)_{k\neq i} \mspace{-5mu}=\mspace{-2mu}\big(x(k,\cdot)\big)_{k\neq i} \right)&  =\frac{1}{2}\sum_{k=1}^{d}  \EE \left[\big(f_j(X)-\tilde{f}_j(X) \big)^2 \middle |\big(X(k,\cdot)\big)_{k\neq i} \mspace{-5mu}=\mspace{-2mu}\big(x(k,\cdot)\big)_{k\neq i}\right] \\
& =\frac{1}{2}\, \EE\left[\sum_{j=1}^d \big(f_j(X) -\tilde{f}_j(X) \big)^2 \middle|\big(X(k,\cdot)\big)_{k\neq i} \mspace{-5mu}=\mspace{-2mu}\big(x(k,\cdot)\big)_{k\neq i} \right] \\
& \leq  \frac{\Delta^2_2(f)}{2}, \numberthis\label{eq:sensitivity l2 DP bound}
\end{align*}
where the last step comes from the definition of $\ell^2$-~sensitivity. Insert \eqref{eq:sensitivity l2 DP bound} into the entropy bound from \eqref{eq:laplace_proof} and use the fact that  $\Delta_2(f) \leq  \Delta_1(f)$ to obtain
\[\mathsf{I}\big(g(X);M_\mathsf{L}(X)\big|w(X)\big)\leq d \log \left(2e  \left(\frac{\Delta_1(f)}{\sqrt{2}} + b \right)\right)- d\log (2be).\]
To conclude, proceed as in the proof of Theorem~\ref{thm:LaplaceMechanismMIPP} to find the parameter $b$.

\subsection{Proof of Theorem \ref{thm:GaussianMechanismMIPP}}\label{proof:GaussianMechanismMIPP}
We follow a similar argument to that in the proof of Theorem~\ref{thm:LaplaceMechanismMIPP}. Denote the conditional covariance matrix of $f(X)$ given $\{w(X)=c\}$ by $\Sigma_{f|w=c}$ and consider
\begin{align*}
\sh\big(f(X)+Z_{\mathsf{G}}\big|w(X)\big)&\stackrel{(a)}\leq  \int { \frac{1}{2} \log \big ((2 \pi e)^d \,\mathrm{det}(\Sigma_{f|w=c}+\sigma^2 \mathrm{I}_d ) \big)} d P_{w(X)}(c) \\ & \stackrel{(b)} \leq \int \frac{1}{2} \log \left( (2 \pi e)^d \prod_{j=1}^d \big (a_j(c)+ \sigma^2\big) \right) d P_{w(X)}(c) \\
& \stackrel{(c)}\leq \frac{d}{2}\int  \log \left( 2 \pi e \left( \frac{1}{d}\sum_{j=1}^d \mathrm{Var}\big(f_j(X) \big|w(X)=c\big)+ \sigma^2\right) \right) d P_{w(X)}(c) \\
& \stackrel{(d)} \leq  \frac{d}{2} \log \left( 2 \pi e \left( \frac{1}{d}\sum_{j=1}^d \EE \Big[\mathrm{Var}\big(f_j(X) \big|w(X)\big)\Big]+ \sigma^2\right) \right) ,\numberthis\label{eq:Gaussian_proof}
\end{align*}
where (a) follows from the Gaussian distribution maximizing differential entropy subject to a variance constant, with $|K|$ denoting the determinant of $K$; (b) denotes $a_j(c):=\mathrm{Var}\big(f_j(X)|w(X)=c\big)$ and uses $|K| \leq \prod_{j=1}^d K(j,j)$, which applies to any positive semi-definite matrix; and (c)-(d) from concavity of $x\mapsto \log x$ and Jensen's inequality. 
 
Combining \eqref{eq:Gaussian_proof} with $\sh\big(f(X)+Z_{\mathsf{G}}\big|g(X),w(X)\big)\geq \sh(Z_{\mathsf{G}})=\frac{d}{2}\log(2\pi e \sigma^2)$ yields 
\[
\mathsf{I}\big(g(X);M_\mathsf{G}(X)\big|w(X)\big)\leq \frac{d}{2} \log \left( 2 \pi e \left( \frac{1}{d}\sum_{j=1}^d \EE \Big[\mathrm{Var}\big(f_j(X) \big|w(X)\big)\Big]+ \sigma^2\right) \right) - \frac{d}{2}\log(2\pi e \sigma^2).
\]
Upper bounding the RHS above by $\epsilon$ and solving for $\sigma^2$  concludes the proof.

\subsection{Proof of Corollary~\ref{cor:GaussianMechanismMIDP}}
We first insert the upper bound from \eqref{eq:sensitivity l2 DP bound} into \eqref{eq:Gaussian_proof} to obtain an $\ell^2$-sensitivity bound on $\sh\big(f(X)+Z_{\mathsf{G}}\big|w(X)\big)$. Combining this with $\sh\big(f(X)+Z_{\mathsf{G}}\big|X\big)\geq \sh(Z_{\mathsf{G}})=\frac{d}{2}\log(2\pi e\sigma^2)$ yields 
\[\mathsf{I}\big(X(i,\cdot);M_\mathsf{G}(X)\big|\big(X(k,\cdot)\big)_{k\neq i}\big)\leq \frac{d}{2} \log \left( 2 \pi e \left(  \frac{\Delta^2_2(f)}{2}+ \sigma^2\right) \right) - \frac{d}{2}\log(2\pi e \sigma^2) \]
Upper bounding the RHS above by $\epsilon$ and solving for~$\sigma^2$ produces the result. 

For the case when $\cX$ is compact and $f:\cX^{n\times k}\to \RR$ is of continuous, we apply the Popoviciu inequality for variance to obtain 
\begin{align*}
    \mathrm{Var}&\left(f_j(X) \middle| \big(X(k,\cdot)\big)_{k\neq i}=\big(x(k,\cdot)\big)_{k\neq i} \right)\\
    &\qquad\qquad\leq \frac{1}{4} \left( \sup_{x(i,\cdot)} f_j\left(x(i,\cdot),\big(x(k,\cdot)\big)_{k\neq i}\right) -\inf_{x(i,\cdot)} f_j\left(x(i,\cdot),\big(x(k,\cdot)\big)_{k\neq i}\right)  \right)^2\\
    &\qquad\qquad\leq \frac{\Delta^2_2(f)}{4}.
\end{align*}
Applying this relation and proceeding with the same argument as above leads to the desired result.

\subsection{Proof of Theorem~\ref{Thm: Gaussian mechanism with slicing and noise injection}}

Invoking Theorem~\ref{thm:GaussianMechanismMIPP} for the query function $g(X)=\rA^\intercal f(X)$, where $\rA=[\phi_1,\ldots, \phi_\ell]$, we obtain
\begin{equation} \label{eq:Gaussian slicing}
     \sigma^2 \geq \sup_{P_X \in \Theta,\, w \in \cW^\star} \frac{ \sum_{j=1}^\ell \EE \sqparen{\mathrm{Var}\big(\phi_j^\intercal f(X)\big|w(X)\big)}}{\ell(e^{\frac{2\epsilon}{\ell}}-1)}. 
\end{equation}
Recall that $\Sigma_{f|w}$ denotes the conditional covariance matrix of $f(X)$ given $w(X)$ by, and bound
\[\EE \sqparen{\mathrm{Var}\big(\phi_j^\tr f(X)\big|w(X)\big)}= \EE \big[\phi_j^\tr \Sigma_{f|w} \phi_j\big]
    \leq \EE\big[\|\Sigma_{f|w}\|_{\mathrm{op}} \|\phi_j\|^2 \big],\quad \forall j=1\,\ldots,\ell.
  \]
Combining the above with \eqref{eq:Gaussian slicing} shows the sufficiency of the variance parameter in Part (1) of the theorem. 

\medskip
For a random projection matrix $\rA=[\Phi_1, \ldots, \Phi_\ell]$, since $\Phi_j$ is centered and independent of $X$, we have $\EE\big[\phi_j^\tr f(X)\big|w(X)\big]=0$. Recalling the notation $\mu_{f|w}:=\EE\big[f(X)|w(X)\big]$, we consequently have 
\begin{align*}
    \EE\big[\mathrm{Var}\big(\Phi_j^\tr f(X)\big|w(X)\big)\big]&= \EE\Big[\Phi_j^\tr \EE\big[ f(X) f(X)^\tr\big|w(X)\big] \Phi_j\Big]\\
    &=\EE\Big[\Phi_j^\tr \big(\Sigma_{f|w} + \mu_{f|w}\mu_{f|w}^\tr\big)\Phi_j\Big]\\
    &\leq\EE\Big[\|\Sigma_{f|w}\|_{\mathsf{op}}+\|\mu_{f|w}\|^2\Big]
\end{align*}
where the last step uses $\EE\big[\|\Phi_j\|^2\big]=1$ and the fact that $\|a a^\tr\|_{\mathsf{op}}\leq \|a\|^2$, for any vector $a\in\RR^d$. Given the variance bound, we proceed as in the proof of the deterministic projection case to obtain the result.

\subsection{Proof of Theorem~\ref{Thm:GaussianMechanismMIPP_with entropy law}}

First, rewrite \eqref{eq:Gaussian_proof} in terms of $A$ given in \cref{Thm:GaussianMechanismMIPP_with entropy law}, to obtain 
\[\sh\big(f(X)+Z_{\mathsf{G}}\big|w(X)\big) \leq \frac{d}{2} \log \paren{2 \pi e \paren{  \frac{A}{d} + \sigma^2}}. \]
From entropy power inequality, we have 
\[ e^{\frac{2}{d}\sh(f(X)+Z_{\mathsf{G}}|g(X),w(X))} \geq e^{\frac{2}{d} \sh(f(X)|g(X),w(X))} + e^{\frac{2}{d} \sh(Z_{\mathsf{G}})}.\]
Noting that $\sh(Z_{\mathsf{G}})=0.5d \log(2 \pi e \sigma^2)$ and by the choice of $B$ in the statement of the theorem, we arrive at
\[\sh\big(f(X)+Z_{\mathsf{G}}\big|g(X),w(X)\big) \geq 0.5 d \log \big(2 \pi e ( B+ \sigma^2)\big),\]
which combined with the above yields
\[\sI\big(g(X);f(X)+Z_{\mathsf{G}}|w(X)\big) \leq \frac{d}{2} \log \paren{2 \pi e \paren{  \frac{A}{d} + \sigma^2}}- \frac{d}{2} \log \big(2 \pi e ( B+ \sigma^2)\big).\]
Upper bounding the RHS by $\epsilon$  and solving for $\sigma^2$ completes the proof. 

\subsection{Proofs related to Remark~\ref{rem:free_MIPP_regime}}\label{sec:Proofs related to Remark free MI PP}

We first show that $A \geq dB$ holds for any $P_X\in\cP(\cX^{n \times k})$ and functions $f$, $g\in\cG$, and $w\in\cW$ with $g\sim w$. Reusing the notation $\Sigma_{f|w}$ for the conditional covariance matrix of $f(X)$ given $w(X)$, we have
\begin{align*}
    \sh\big(f(X)|w(X),g(X)\big) &\leq \sh\big(f(X)|w(X)\big) \\
    &\stackrel{(a)}\leq \frac{1}{2}\, \EE \Big[ \log \left( (2\pi e)^d \mathrm{det}(\Sigma_{f|w})\right)\Big]\\
    & \stackrel{(b)}\leq \frac{1}{2}\, \EE \left[ \log \left( (2\pi e)^d  \prod_{j=1}^d \mathrm{Var}\big( f_j(X)\big|w(X)\big)\right)\right] \\
    & \leq \frac{d}{2}\,  \EE \left[ \log \left( \frac{2\pi e}{d}  \sum_{j=1}^d\mathrm{Var}\big( f_j(X)|w(X)\big)\right)\right] \\
    & \leq \frac{d}{2}\,  \log \left( \frac{2\pi e}{d}  \sum_{j=1}^d \EE \left[ \mathrm{Var}\big( f_j(X)|w(X)\big)\right]\right)
\end{align*}
where (a) is since the Gaussian distribution maximizing entropy under finite variance constraint; (b) uses $|K| \leq \prod_{j=1}^d K(j,j)$, which applies to any positive semidefinite matrix;
and the last two steps follow from Jensen's inequality. Substituting the above bound in place of $B$ in Theorem~\ref{Thm:GaussianMechanismMIPP_with entropy law} yields the result. 

\medskip 
Next, we show that when $d=1$ and $\big(f(X),g(X),w(X)\big)$ are jointly Gaussian, we indeed have $A \leq d \, e^{2\epsilon/d} B$ under the said correlation coefficient bound. Under this Gaussian setting, we have 
\[ A= \EE\big[ \mathrm{Var}\big(f(X)|w(X)\big)\big] = \mathrm{Var}\big(f(X)|w(X)=c\big),\quad \forall c \in \mathrm{Im}(w).\]
Similarly, joint Gaussianity of the involved variables implies 
\[ B= \mathrm{Var}\big(f(X)|w(X)=c, g(X)=a\big),\quad \forall(a,c) \in \mathrm{Im}({g})\times \mathrm{Im}(w).\]
Inserting this into the inequality $A \leq d \, e^{2\epsilon/d} B$ with $d=1$, while observing that 
\[\mathrm{Var}\big(f(X)\big|w(X)=c, g(X)=a\big)=\left( 1-  \left(\rho\left(f(X),g(X)\big|w(X)=c \right) \right)^2  \right) \mathrm{Var}\big(f(X)|w(X)=c\big) \]
completes the proof.

\subsection{Proof of Proposition~\ref{Prop:High dimensional mean estimation accuracy guarantees} } \label{Proof of Prop:High dimensional mean estimation accuracy guarantees}

We show that the estimate produced by Algorithm~\ref{alg:d large MI DP} satisfies $\epsilon$-MI DP, and establish that given $n \geq n_0$ samples, the estimator also achieves accuracy $\alpha$ with high probability. Denote $c:=\EE\big[\|X-\mu\|^2\big] < \infty$ and notice that the noisy mean estimates $(\tilde{\mu}_1,\ldots,\tilde{\mu}_m)$ produced during the $m$ rounds of Algorithm~\ref{alg:d large MI DP} are i.i.d. Further assume that $n=mk$ (otherwise, simple modifications of the subsequent argument using ceiling/floor operation may be needed). Furthermore denote $A_p:=\{(p-1)k+1,\ldots, pk \}$ for $p=1,\ldots,m$, and note that $\bigcup_{p=1}^m A_p=\{1, \ldots, n \}$. We first establish privacy of the mechanism and then analyze its accuracy. 

\medskip
\underline{Privacy analysis:} For privacy, we first show that each $\tilde{\mu}_p$, $p=1,\ldots,m$, is $\epsilon$-MI DP. Having that, we argue that $M^m(X)=(\tilde{\mu}_1, \ldots \tilde{\mu}_m)$ is private in the same sense via composition, and finally deduce the privacy of the geometric median $\hat{\mu}_n$ via post-processing (Property (2) of \cref{thm:MIPP_properties}). Consider $\tilde{\mu}_1$ and notice that it satisfies $\epsilon$-MI DP by Corollary~\ref{cor:GaussianMechanismMIDP}. Indeed, for each $i=1,\ldots,k$ and $\ell=1,\ldots, d$, we have
\[\mathrm{Var}\left(\tilde{\mu}_1 (\ell)\middle| \big(X_j\big)_{j \in A_1 \backslash \{i\}}\right) = \frac{1}{k^2} \mathrm{Var}\big(X_i(\ell)\big) \leq \frac{m^2 c}{n^2}, \]
where the first equality is since $X_1,\ldots,X_k$ are i.i.d. and the second uses the 2nd moment bound. Therefore,
\[\frac{\sum_{\ell=1}^d\mathrm{Var}\left(\tilde{\mu}_1 (\ell)\middle| \big(X_j\big)_{j \in A_1 \backslash \{i\}}\right)}{d(e^{\frac{2\epsilon}{d}}-1)}  \leq \frac{d m^2 c}{2n^2 \epsilon}, \]
where we have used $e^x\geq 1+x$. By  Corollary~\ref{cor:GaussianMechanismMIDP} we now see that the noise level stated in Algorithm~\ref{alg:d large MI DP} suffices to guarantee $\epsilon$-MI DP of $\tilde{\mu}_1$. By symmetry, the same hold for all $\tilde{\mu}_p$, $p=1,\ldots,m$.

Next, we show that $(\tilde{\mu}_1, \ldots \tilde{\mu}_m)$ is $\epsilon$-MI DP w.r.t. the entire database $(X_1,\ldots,X_n)$. Fix $i=1,\ldots,n$ and let $p(i)\in\{1,\ldots,m\}$ be such that $i \in A_{p(i)}$. Then consider
\begin{align*}
    \sI\big (X_i; \tilde{\mu}_1, \ldots ,\tilde{\mu}_m \big|(X_j)_{j\neq i} \big) &=  \sI\big(X_i; \tilde{\mu}_{p(i)} \big|\big(X_j\big)_{j\neq i} \big) + \sI\big(X_i; ( \tilde{\mu}_q)_{q\neq p(i)} \big|(X_j)_{j\neq i},\tilde{\mu}_{p(i)} \big) \\ 
    & \stackrel{(a)}= \sI\big(X_i; \tilde{\mu}_{p(i)} \big|(X_j)_{j\in A_{p(i)}\backslash \{i\}} \big)\\
   & \stackrel{(b)} \leq  \epsilon
\end{align*}
where (a) follows since 
$(X_i,\tilde{\mu}_{p(i)}) \leftrightarrow (X_j)_{j\in A_{p(i)} \backslash\{i\}} \leftrightarrow (X_j)_{j \notin A_{p(i)}}$ and $ X_i \leftrightarrow \big( (X_j)_{j\neq i},\tilde{\mu}_{p(i)} \big) \leftrightarrow (\tilde{\mu}_q)_{q\neq {p(i)}}$ form Markov chains, while (b) is since $\tilde{\mu}_{p(i)}$ satisfies $\epsilon$-MI DP. 
Maximizing the LHS above over $i=1,\ldots,n$ yields 
\[ \sup_{i=1,\ldots,n} \sI\big (X_i; \tilde{\mu}_1, \ldots ,\tilde{\mu}_m \big|(X_j)_{j\neq i} \big) \leq \epsilon, \]
which shows that $(\tilde{\mu}_1, \ldots ,\tilde{\mu}_m) $ is $\epsilon$-MI DP.
Recalling that post-processing preserves $\epsilon$-MI DP (cf. Property (3) of Theorem~\ref{thm:MIPP_properties}), we conclude that the geometric median of $(\tilde{\mu}_1, \ldots ,\tilde{\mu}_m)$ also satisfies $\epsilon$-MI DP.

\medskip
\underline{Accuracy analysis:} We first derive a lower bound on $k$ (the number of samples used to evaluate $\tilde{\mu}_p$ for each $p=1, \ldots, m$) such that the mean estimate $\tilde{\mu}_p$ is closer to true mean $\mu$ with high probability. Then, we invoke \cite[Theorem 3.1]{minsker2015geometric} on confidence boosting via geometric medians to 
argue that the geometric median of these $m$ estimates satisfies the accuracy level of accuracy stated in the theorem.

The following lemma states the lower bound on $k$ for a single iteration of Algorithm~\ref{alg:d large MI DP}, i.e., for each $1,\ldots,m$.
\begin{lemma}[Accuracy of a single iteration of Algorithm~\ref{alg:d large MI DP}] \label{lemma:Accuracy of one iteration d }
Fix $p=1,\ldots,m$ and suppose that \[ k \geq O\left(  \frac{d}{\alpha'^2}  + \frac{d}{\alpha' \sqrt{\epsilon}}   \right). \]
Then the mean estimate $\tilde{\mu}_p$ produced by the $p$th iteration of Algorithm~\ref{alg:d large MI DP} satisfies $\PP\big(\|\tilde{\mu}_p - \mu\| \leq \alpha'\big) \geq 0.8$. 
\end{lemma}
\begin{proof}
We first bound the empirical mean estimation error. By Chebyshev's inequality and since the second moment is bounded by $c$, we have
\begin{align*}
    \PP\left(\left\|\mu- \frac{1}{k} \sum\nolimits_{i \in A_p} X_i \right \| > \frac{\alpha'}{2}  \right) &\leq \frac{4\EE\left[\big\|\mu- \frac{1}{k} \sum_{i \in A_p} X_i  \big \|^2 \right]}{\alpha'^2} \\
    &\leq \frac{4d c}{k {\alpha'}^2},
\end{align*}
Setting $k \geq 4dc/(0.1{\alpha'}^2)$ guarantees that $\big\|\mu- \frac{1}{k} \sum_{i \in A_p} X_i  \big \| \leq {\alpha'}/{2}$ at least with probability 0.9. 

Next, consider the error due to noise injection for privacy with parameter $\sigma^2= dc/(2k^2 \epsilon)$. Using the tail bounds for $d$-dimensional Gaussian vector we have 
\[ \PP\left( \|Z_p\| > \frac{\alpha'}{2} \right) 
\leq 2\exp\left(-\frac{\alpha'^2}{8 \sigma^2 d}\right).\] 
Choosing $k \geq \frac{2dc}{\alpha \sqrt{\epsilon}} \sqrt{\log(20)}$ guarantees that $\|Z_p  \| \leq \frac{\alpha'}{2}$  at least with probability 0.9.

The choice of $k$ as stated in the Lemma is sufficient to satisfy both bounds. Now consider
\begin{align*}
    \PP\big(\|\tilde{\mu}_p - \mu\| \leq \alpha'\big) &  \geq \PP\left(\left\|\mu-\frac{1}{k} \sum\nolimits_{ i \in A_p} X_i \right \| \leq \frac{\alpha'}{2}, \|Z_p  \| \leq  \frac{\alpha'}{2} \right)\\
    &\geq 1- \PP\left(\left\|\mu-\frac{1}{k} \sum\nolimits_{i \in A_p} X_i \right \| >  \frac{\alpha'}{2}  \right)-\PP\left( \|Z_p  \| >\frac{\alpha'}{2}  \right)
\end{align*}
The result follows by recalling that each term on the RHS is less than 0.1.
\end{proof}

From Lemma~\ref{lemma:Accuracy of one iteration d }, we have  $\PP\big(\|\tilde{\mu}_p - \mu\| \geq \alpha'\big) \leq 0.2$ for all $p=1,\ldots,m$, so long that $k$ satisfies the prescribed lower bound. Applying Theorem~3.1 of~\cite{minsker2015geometric} on boosting the confidence via geometric medians, with our choice of $m=200 \log(1/\beta)$ yields\footnote{Specifically, adapting to their notation, we invoke \cite[Theorem~3.1]{minsker2015geometric} with $p=0.2$, $\epsilon=\alpha'$, and $\alpha=0.21$.}
\[ \PP\big(\|\hat{\mu}_n - \mu\| \geq 1.04 \alpha'\big) \leq \beta.\] 
Setting $\alpha= 1.04 \alpha'$ and noticing that $n= k m $ completes the proof of the accuracy guarantee.

\section{Concluding Remarks and Future directions} \label{Sec: Conclusion}
This work established an information-theoretic characterization, termed $\epsilon$-MI PP, of the structured PP framework, where private information is modeled by functions of the database. The characterization was leveraged to derive properties of $\epsilon$-MI PP and obtain sufficient conditions for noise-injection (Laplace and Gaussian) mechanisms. Our results highlight the virtues of an information-theoretic perspective on PP, enabling more flexible composability theorems and variance-dependent noise parameter bounds that exploit the distributional assumptions in the PP framework. As applications of $\epsilon$-MI PP we explored auditing privacy frameworks, statistical inference tasks, and privacy-utility tradeoffs. 

Future research directions are abundant. First, we aim to better and strengthen the reverse implication in \cref{thm:equivalence}, i.e., derive relaxed and general conditions under which $\epsilon$-MI PP implies PP in the classic sense. This is a non-trivial endeavour as mutual information is an average quantity, while classic privacy notions are typically worst-case. We also target a power (namely, Type II error) analysis of the auditing hypothesis test in \cref{Sec: Auditing}. A possible direction from which to tackle this question is to derive a limit distribution theory under the alternative for the test statistic and use that to analyze the power of the test under local alternatives via LeCam's Third Lemma. Lastly, we are interested in further exploring to privacy-utility tradeoffs  \cite{zamani2022bounds,PrivacyUtilityTotalVariation,PrivacyUUtilityActiveHypothesis} via $\epsilon$-MI PP and connect those to tradeoffs for standard PP mechanisms (preliminary results in this direction are found in \cref{app:utility}). In particular, we aim to characterize the achievable privacy-utility region for PP mechanisms and design optimal mechanisms for different points in that region.

\section*{Acknowledgment}

The authors thank Rachel Cummings for her feedback and helpful suggestions regarding an earlier version of the $\epsilon$-MI PP idea.

 \bibliographystyle{ieeetr}
\bibliography{reference}

\end{document}